\documentclass[11pt, letterpaper]{article}

\usepackage{vasilis}
\newcommand{\inv}{\leftarrow}

\allowdisplaybreaks

\title{Minimization I.I.D. Prophet Inequality via Extreme Value Theory: A Unified Approach}

\author{
Vasilis Livanos
\thanks{EPFL, Lausanne, Switzerland \& Archimedes AI, Athens, Greece {\tt (vasileios.livanos@epfl.ch)}. Work on this paper partially supported by NSF grant CCF-1750436.}
\and
Ruta Mehta
\thanks{Dept.\ of Computer Science, University of Illinois at Urbana-Champaign, IL 61801. {\tt rutameht@illinois.edu}}
}

\date{}

\begin{document}
\maketitle
\thispagestyle{empty}

\begin{abstract}
The \emph{I.I.D. Prophet Inequality} is a fundamental problem in optimal stopping theory where, given $n$ independent random variables $X_1,\dots,X_n$ drawn from a known distribution $\cD$, one has to decide at every step $i$ whether to stop and accept $X_i$ or discard it forever and continue. The goal is to maximize (or minimize) the selected value and compete against the all-knowing prophet. For the maximization setting, a tight constant-competitive guarantee of $\approx 0.745$ is well-known (Correa, Foncea, Hoeksma, Oosterwijk, Vredeveld, 2019), whereas the minimization setting is qualitatively different: the optimal constant is distribution-dependent and can be arbitrarily large (Livanos and Mehta, 2024).

In this paper, we provide a novel framework via the lens of \emph{Extreme Value Theory} to analyze optimal threshold algorithms. We show that the competitive ratio for the minimization setting has a \emph{closed form} described by a function $\Lambda$, which depends only on the \emph{extreme value index} $\gamma$; in particular, it corresponds to $\Lambda(\gamma)$ for $\gamma \leq 0$. Despite the contrast of the optimal guarantees for maximization and minimization, our framework turns out to be universal and is able to recover the results of (Kennedy and Kertz, 1991) for the maximization case as well. Surprisingly, the optimal competitive ratio for the maximization setting is given by the same function $\Lambda(\gamma)$, but for $\gamma \geq 0$. Along the way, we obtain several results on the algorithm and the prophet's objectives from the perspective of extreme value theory, which might be of independent interest.


We next study \emph{single-threshold} algorithms for the minimization setting. Again, using the extreme value theory, we generalize the results of (Livanos and Mehta, 2024) which hold only for special classes of distributions, and obtain poly-logarithmic in $n$ guarantees on the competitive ratio. Finally, we consider the $k$-multi-unit prophet inequality in the minimization setting and show that there exist constant-competitive single-threshold algorithms when $k \geq \log{n}$.

\end{abstract}

\medskip
\medskip
\medskip
\noindent

\clearpage
\setcounter{page}{1}

\section{Introduction}\label{sec:introduction}

The {\em prophet inequality} is the following classical problem in optimal stopping theory, first introduced by Krengel, Sucheston and Garling in 1977 \cite{kren-such, kren-such2}:  
one is presented with \emph{take-it-or-leave-it} rewards $X_1, \dots, X_n$ in an online manner, where each $X_i$ is drawn from a known distribution $\cD_i$, independently from the other rewards, and can stop at any point and collect the last reward seen. The inequality ensures the existence of a \emph{stopping strategy} $S$ for any arrival order of the random variables, with expected reward at least half that of a \emph{prophet} who can see the realizations of all the $X_i$'s upfront and thus can always select the maximum, {\em i.e.,}  $\E\brk{S} \geq \f{1}{2} \E\brk{\max_i X_i}$. The \emph{competitive ratio} $\frac{\E\brk{S}}{\E\brk{\max_i X_i}}$ quantifies the loss that the algorithm incurs by observing the realizations sequentially instead of all at once. A simple example shows that the factor of $\f{1}{2}$ is tight in this general setting.

A special case of the prophet inequality that is of significant interest is the case where the random variables are \emph{independent and identically distributed (I.I.D.)}, i.e. $\cD_i = \cD$ for all $i$. For this setting, Hill and Kertz \cite{hill-kertz} showed the existence of a simple threshold strategy that is optimal. They showed it achieves a competitive ratio of $1 - \f{1}{e}$, as well as providing a corresponding upper bound of $\approx 0.745$. Subsequent work \cite{abolh} improved the analysis of the optimal algorithm and showed a lower bound of $0.738$, before Correa, Foncea, Hoeksma, Oosterwijk and Vredeveld \cite{correa-iid} obtained the tight factor of $0.745$.

These results, and their variations and generalizations, have found extensive applications. Starting with the work of Hajiaghayi, Kleinberg and Sandholm \cite{haji} and later Chawla, Hartline, Malec and Sivan \cite{ChawlaHMS}, the prophet inequality has been connected to the design of simple yet approximately optimal sequential posted-price mechanisms. Due to the significance of such problems in Bayesian mechanism design, the study of prophet inequalities has seen a tremendous surge in the last decade \cite{klein-wein,EhsaniHKS18,Dutting2,dutting-combinatorial,esf-prophsec,azar,correa-random,correa-iid,renato-iid,correa,abolh,CorreaPricing,ma-tight-k-prophet,willma-random,rs,chek-liv,better-tightness,sample-pi-1,sample-pi-2,sample-pi-3}.

Recently, Livanos and Mehta \cite{liv-mehta-cpi} initiated the study prophet inequalities for \emph{cost minimization}. They study the I.I.D. setting, which is similar to the classical prophet inequality with the only differences being that instead of trying to maximize the selected value, the algorithm aims to select a value as small as possible and compares against a prophet who can always select the minimum and obtain $\E\brk{\min_i X_i}$. Furthermore, in this setting, both the algorithm and the prophet are forced to select a realization -- in other words, the constraint is upwards-closed instead of downwards-closed -- and the competitive ratio is always at least $1$. \cite{liv-mehta-cpi} shows that even for $n = 2$ I.I.D. random variables drawn from a fixed distribution, the competitive ratio can be unbounded. This is in stark contrast with the maximization setting where a constant competitive ratio is known even when the random variables are drawn from different distributions. Livanos and Mehta then study the optimal strategy, which is given by a dynamic programming algorithm, and show that it achieves a (distribution-dependent) constant competitive ratio for a special class of distributions they call \emph{entire} distributions. They provide a closed form for this competitive ratio, but their results fail to generalize to all distributions. 

The known analysis framework and techniques for maximization and minimization settings differ significantly. On the other hand, the two settings seem dual-like, raising a natural question: 

\begin{center}
{\em Does there exist a unified framework through which we can obtain the optimal approximation ratios for both the minimization and the maximization setting, despite their apparent difference?}
\end{center}

In this paper, we answer the above question affirmatively. We design an analysis framework via the lens of \emph{Extreme Value Theory (EVT)} for the \emph{large market} case, where the distribution $\cD$ is independent of $n$, and we investigate the competitive ratio's behaviour as $n$ goes to infinity, for both the maximization and minimization settings. We call this the \emph{Asymptotic Competitive Ratio (ACR)}. Next, we discuss basic concepts from Extreme Value Theory, before continuing with a brief overview of our results and techniques.

\subsection{Extreme Value Theory}\label{sec:evt}

Our analysis assumes that $\cD$ belongs to a special class of distributions, for which maxima and minima of a sample converge in distribution. However, this assumption is quite mild since the set of such distributions is dense \cite{evt-in-general}. Furthermore, one cannot hope to obtain a unified analysis for a more general class of distributions, since the maxima and minima of other distributions do not even converge in distribution and thus even the prophet's objective defies a closed-form description. For an introduction to Extreme Value Theory see \cite{dehaan-ferreira-evt}.

Convergence of the maxima and minima of a distribution is captured by the celebrated Extreme Value Theorem, also known as the Fisher-Tippett-Gnedenko Theorem. Recall that for a distribution $\cD$ with cdf $F(x)$, the cdf of the maximum of $n$ samples is $F^n(x)$ and the cdf of the minimum of $n$ samples is $1 - \prn{1 - F(x)}^n$.

\begin{theorem}[Extreme Value Theorem (\cite{fisher-tippett-evt, gnedenko-evt})]\label{thm:gnedenko}
Let $X_1, \dots, X_n$ be a sequence of I.I.D. random variables with cumulative distribution function $F$. Suppose that there exist two sequences $a_n > 0, b_n \in \R$ such that the following limit converges to a non-degenerate distribution function
\[
\lim_{n \to \infty} {F^n(a_n x + b_n)} = G_\gamma(x).
\]
Then, $G_\gamma(x)$ must be of the form
\[
G_\gamma(x) = \begin{cases}
\exp\prn{-\prn{1+\gamma x}^{-\f{1}{\gamma}}}, & \text{if } \gamma \neq 0, \\
\exp\prn{-\exp\prn{-x}}, & \text{if } \gamma = 0,
\end{cases}
\]
for all $1+\gamma x > 0$.

Similarly, suppose that there exist two sequences $a'_n > 0, b'_n \in \R$ such that the following limit converges to a non-degenerate distribution function
\[
\lim_{n \to \infty} {\prn{1 - F(a_n x + b_n)}^n} = G^*_\gamma(x).
\]
Then $G^*_\gamma(x)$ must be of the form
\[
G^*_\gamma(x) = \begin{cases}
\exp\prn{-\prn{1-\gamma x}^{-\f{1}{\gamma}}}, & \text{if } \gamma \neq 0, \\
\exp\prn{-\exp\prn{x}}, & \text{if } \gamma = 0.
\end{cases}
\]
for all $1-\gamma x > 0$.
\end{theorem}

In the theorem above, $\gamma$ is called the \emph{extreme value index}, and it partitions the space of all distributions of the theorem into equivalence classes based on which $G_\gamma$ their maxima and minima converge to.
\begin{itemize}
\item For $\gamma > 0$, we obtain the \emph{Fr\'{e}chet class} of distributions.
\item For $\gamma = 0$, we obtain the \emph{Gumbel class} of distributions.
\item For $\gamma < 0$, we obtain the \emph{Reverse Weibull class} of distributions.
\end{itemize}

One can think of the Extreme Value Theorem as the analogue of the Central Limit Theorem (CLT) for maxima and minima instead of averages; in the same way that the CLT ensures that the (properly scaled) average of $n$ random variables drawn independently from the same distribution $\cD$ converges in distribution to a standard Gaussian distribution, as long as their variance is finite, the Fisher-Tippett-Gnedenko Theorem ensures that, if the (properly scaled) maxima and minima of $n$ random variables drawn independently from $\cD$ converge in distribution, then they can only converge to a distribution of the form of $G_\gamma$ for maxima $G^*_\gamma$ for minima.

For every distribution $\cD$ of Theorem~\ref{thm:gnedenko}, we can associate a single extreme value index $\gamma$. We then say that $\cD$ belongs in the \emph{domain of attraction} of $D_\gamma$.

\begin{definition}[Domain of Attraction]\label{def:domain-of-attraction}
When considering the I.I.D. Max-Prophet Inequality, we say that a distribution $\cD$ with cdf $F$ belongs in the domain of attraction of $D_\gamma$, indicated as $\cD \in D_\gamma$, if there exist two sequences $a_n > 0, b_n \in \R$ such that the following limit converges to a non-degenerate distribution function
\[
\lim_{n \to \infty} {F^n(a_n x + b_n)} = G_\gamma(x).
\]
Similarly, when considering the I.I.D. Min-Prophet Inequality, we say that a distribution $\cD$ with cdf $F$ belongs in the domain of attraction of $D_\gamma$, indicated as $\cD \in D_\gamma$, if there exist two sequences $a'_n > 0, b'_n \in \R$ such that the following limit converges to a non-degenerate distribution function
\[
\lim_{n \to \infty} {\prn{1 - F(a'_n x + b'_n)}^n} = G^*_\gamma(x).
\]
\end{definition}

For simplicity, we usually write $F \in D_\gamma$ to indicate that the distribution $\cD$ which has cdf $F$ belongs in the domain of attraction of $D_\gamma$, and we say that $\cD$ satisfies or follows the EVT. Whether we refer to the I.I.D. Max-Prophet Inequality or the I.I.D. Min-Prophet Inequality will be clear from context.

To assist the reader, we provide an example of a standard distribution and its domains of attraction for maxima and minima.
\begin{example}
Let $F_U$ denote the uniform distribution, supported on $[0,1]$. Clearly, $F_U(x) = x$, for all $x \in [0,1]$. Let $a_n = 1/n$ and $b_n = 1 - 1/n$. Notice that
\[
\lim_{n \to \infty} F^n_U(a_n x + b_n) = \lim_{n \to \infty}  \prn{1 + \frac{x-1}{n}}^n = e^{x - 1} = G_{-1}(x).
\]
Thus, the uniform distribution follows the EVT for maxima with $\gamma = -1$. Similarly, let $a'_n = 1/n$ and $b'_n = 1/n$. Then,
\[
\lim_{n \to \infty} {\prn{1 - F_U(a'_n x + b'_n)}^n} = \lim_{n \to \infty} {\prn{1 - \frac{x+1}{n}}^n} = e^{-(1+x)} = G^*_{-1}(x).
\]
Thus, the uniform distribution also follows the EVT for minima with $\gamma = -1$\footnote{We note that, in general, the extreme value index $\gamma$ for maxima and minima need not be the same; for example, the exponential distribution with cdf $F(x) = 1 - e^{-x}$ follows the EVT for maxima with $\gamma = 0$ but for minima with $\gamma = -1$.}.
\end{example}

Before we proceed, we remark on our assumption of $\cD$ being a continuous distribution. The class of distributions that obey the Extreme Value Theorem contains both continuous and discrete distributions\footnote{A couple of examples of discrete distributions that obey the EVT are the discrete analogue of the Pareto distribution, which belongs to the domain of
attraction of the Fr\'{e}chet distribution for maxima, as well as the discrete analogue of the log-normal distribution, which belongs to the domain of attraction of the Gumbel distribution \cite{first-course-order-statistics}.}. However, several discrete distributions of significant interest, such as the Poisson or the geometric, fail to obey the EVT. To capture these, Shimura \cite{shimura-discrete-evt} studied the class of distributions that arise by discretizing a continuous distribution. He considered a sequence $\set{X_n}_{n \in \N}$ of random variables drawn I.I.D. from a continuous $\cD$ that obeys the EVT and is in the domain of attraction of the Gumbel or the Reverse Weibull distribution (i.e. for $\gamma \leq 0$). He then gave a necessary and sufficient condition for the discretized sequence $Y_n = \floor{X_n}$ to obey the EVT. Via this approach, he showed that the EVT can be extended to discrete distributions and interesting discrete distributions such as the Poisson or the geometric can be recovered as discretizations of continuous distributions obeying the EVT. Later, Hitz, Davis and Samorodnitsky \cite{discrete-evt-frechet} tackled the Fr\'{e}chet class of distributions, giving an alternative limiting (parametric) distribution for the discrete case: the \emph{Zipf-Mandelbrot distribution} (\textsc{ZMD}) \cite{mandelbrot-distribution}. They showed that, for every discrete distribution $\cD$ in the domain of attraction of the ZMD with parameter $\gamma / (1 + \gamma)$ where $\gamma \geq 0$, every continuous distribution $\cD'$ that agrees with $\cD$ on its discrete support has to belong to the domain of attraction of $D_\gamma$, i.e. $\cD'$ satisfies the EVT with parameter $\gamma$. For the remainder of our paper, we assume that $\cD$ is a continuous distribution. However, by the discussion above, our results extend to discrete distributions that satisfy the EVT (e.g. the discrete analogue of the Pareto distribution) or are discretizations of continuous distributions that satisfy the EVT (e.g. the Poisson or geometric distributions); we refer the reader to \cite{shimura-discrete-evt,discrete-evt-frechet} for further information on the extensions of the EVT for discrete distributions.

\subsection{Our Contributions and Techniques}\label{sec:results}

We study the asymptotic competitive ratio of the I.I.D. prophet inequality setting for both the optimal (dynamic programming) algorithm and the optimal single-threshold algorithm. We make three main contributions: (i) we obtain a {\em unified analysis} of the ACR for the optimal strategy, for both maximization and minimization settings, (ii) we study single-threshold algorithms for the minimization setting, showing that the optimal ACR of single-threshold algorithms is poly-logarithmic in $n$, and (iii) we study the $k$-multi-unit minimization prophet inequality, in which the algorithm has to select at least $k$ values, and show that there exist constant-competitive single-threshold algorithms when $k \geq \log{n}$.

Importantly, we note that Kennedy and Kertz \cite{kennedy-kertz} first investigated the ACR of the optimal strategy for the maximization I.I.D. prophet inequality for distributions that follow the EVT. They showed that for any distribution $F \in D_\gamma$, the ACR only depends on $\gamma$, and gave a closed form. Later, Correa, Pizarro and Verdugo \cite{jose-dana-victor-evt} studied the ACR of single-threshold algorithms for distributions that follow EVT. They also showed that the ACR of single-threshold algorithms depends only on $\gamma$ and gave a closed form for it. However, the literature is scarce in results for the minimization I.I.D. prophet inequality. Surprisingly, our first result shows that although the minimization and maximization problems are qualitatively poles apart, they can be connected through a unified analysis via extreme value theory. It provides an elegant closed-form characterization of the ACR for minimization and, at the same time, recovers the \cite{kennedy-kertz} result for maximization.

\paragraph{Optimal Threshold Algorithm.}
Our main contribution is to show that the behaviour of the optimal ACR for every distribution that satisfies the EVT theorem is governed by a unique function $\Lambda$ for both the maximization and minimization settings, which depends only on the \emph{extreme value index}. Interestingly, the positive domain of $\Lambda$ corresponds to the ACR for the maximization setting -- and is equal to the closed form of Kennedy and Kertz \cite{kennedy-kertz} -- and the negative domain of $\Lambda$ corresponds to the ACR for the minimization setting. Our main technical ingredient is a unified analysis of both settings via the lens of extreme value theory and the theory of regularly varying functions, through which we recover, simplify and generalize the results of \cite{kennedy-kertz} for the maximization case, which they obtained via different techniques. Our analysis also provides several technical results on the algorithm's and prophet's objectives, which might find further use.

Let $\Gamma(x)$ denote the Gamma function, $G_M(n)$ (resp. $G_m(n)$) denote the expected value of the optimal threshold algorithm for the I.I.D. Max-Prophet Inequality (resp. I.I.D. Min-Prophet Inequality) setting with $n$ random variables. Furthermore, let $\lambda_M = \lim_{n\to +\infty} \frac{G_M(n)}{\E\brk{\max_{i = 1}^n X_i}} \Bigl($resp. $\lambda_m = \lim_{n\to +\infty} \frac{G_m(n)}{\E\brk{\min_{i = 1}^n X_i}}\Bigr)$ denote the asymptotic competitive ratio in the I.I.D. Max-Prophet Inequality (resp. I.I.D. Min-Prophet Inequality) setting. For formal definitions and more information, see Section~\ref{sec:preliminaries}. What follows is our main theorem.

\begin{theorem}\label{thm:unified-acr}
Let $F \in D_\gamma$. Then,
\begin{itemize}
    \item for the I.I.D. Max-Prophet Inequality and every $\gamma \in \R$, the asymptotic competitive ratio of the optimal threshold policy as $n \to \infty$ is
    \[
    \lambda_M= \min\set{\frac{\prn{1 - \gamma}^{-\gamma}}{\Gamma\prn{1 - \gamma}}, 1}.
    \]
    \item for the I.I.D. Min-Prophet Inequality, we have that $\gamma \leq 0$, and the asymptotic competitive ratio of the optimal threshold policy as $n \to \infty$ is
    \[
    \lambda_m = \max\set{\frac{\prn{1 - \gamma}^{-\gamma}}{\Gamma\prn{1 - \gamma}}, 1}.
    \]
\end{itemize}
\end{theorem}

We briefly discuss the unique function that characterizes the ACR. Let $\Lambda(\gamma) \triangleq \frac{\prn{1 - \gamma}^{-\gamma}}{\Gamma(1 - \gamma)}$. $\Lambda$ can be seen in Figure~\ref{fig:acr}.
\begin{figure}[ht]
        \centering
        \begin{minipage}{\textwidth}
            \centering
            \includegraphics[width=0.3\linewidth]{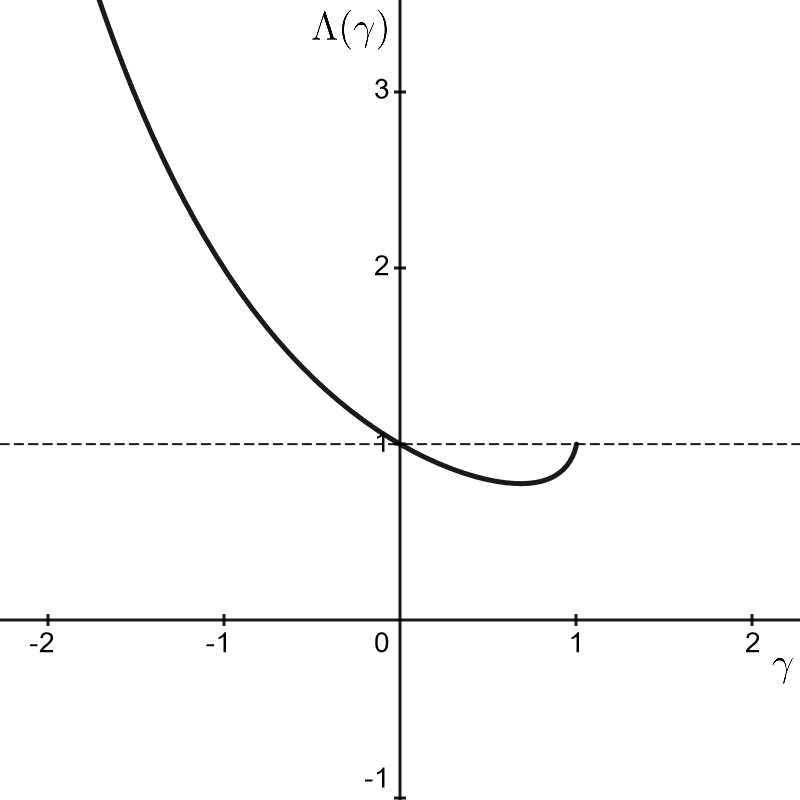}
            \caption{$\Lambda(\gamma)$}
            \label{fig:acr}
        \end{minipage}%
\end{figure}

Our main theorem essentially states that the interesting cases of the ACR for the I.I.D. Max-Prophet Inequality and the I.I.D. Min-Prophet Inequality, in which it is not equal to $1$, are captured by the positive and negative domains of $\Lambda$, respectively. i.e. $\lambda_M = \Lambda(\gamma)$ for $\gamma \geq 0$ in the I.I.D. Max-Prophet Inequality and $\lambda_m = \Lambda(\gamma)$ for $\gamma \leq 0$ in the I.I.D. Min-Prophet Inequality.

In the I.I.D. Max-Prophet Inequality case, we have that $\Lambda(\gamma)$ for $\gamma \geq 0$ is always at least $0.776$. Thus, as a corollary, we obtain that the asymptotic competitive ratio of the I.I.D. Max-Prophet Inequality is at least $\approx 0.776$, which is slightly larger than the factor of $\approx 0.745$ obtained by \cite{correa-iid} (where the distribution can depend on $n$). To understand how $\Lambda(\gamma)$ grows as $\gamma \to -\infty$ in the I.I.D. Min-Prophet Inequality case, consider Stirling's approximation for the Gamma function,
$
\Gamma(z) \approx \frac{\sqrt{2 \pi}}{z} \prn{\frac{z}{e}}^z.
$
Replacing this in the expression of $\Lambda(\gamma)$, we have
\[
\Lambda(\gamma) = \frac{\prn{1-\gamma}^{-\gamma}}{\Gamma\prn{1-\gamma}} \approx 
\frac{\prn{1-\gamma}^{-\gamma}}{\frac{\sqrt{2 \pi}}{1 - \gamma} \prn{\frac{1-\gamma}{e}}^{1-\gamma}} = \bigTh{e^{-\gamma}}.
\]
Thus, the dependence of $\Lambda$ on $\gamma$ is (approximately) exponential, for $\gamma \to -\infty$.

\paragraph{Techniques.}
Our approach to show Theorem~\ref{thm:unified-acr} relies on three steps. For simplicity, assume below that $\gamma > 0$ for the Max-Prophet Inequality and $\gamma < 0$ for the Min-Prophet Inequality; for all other values of $\gamma$, similar results can be obtained. We use $\approx$ to denote asymptotic equality -- see Section~\ref{sec:preliminaries} for more details.
\begin{enumerate}
    \item[(i)] First, we show that the prophet's value, i.e. the expected maximum or minimum, can be (approximately) expressed as the $n$-th quantile of the distribution (Lemma~\ref{lem:mu-approx}), with a multiplicative factor that depends on $\gamma$; specifically:
    \begin{equation}\label{eq:intro-1}
    \E\brk{\max_{i = 1}^n X_i} \approx \Gamma(1 - \gamma) F^{\inv}\prn{1 - \frac{1}{n}} \text{ and } \E\brk{\min_{i = 1}^n X_i} \approx \Gamma(1 - \gamma) F^{\inv}\prn{\frac{1}{n}}.
    \end{equation}

    \item[(ii)] Next, we aim to express the expected value of the optimal threshold algorithm as a quantile of the distribution (proof of Theorem~\ref{thm:unified-acr}), in order to relate it to the expression above:
    \begin{equation}\label{eq:intro-2}
    G_M(n) \approx F^{\inv}\prn{1 - \frac{1-\gamma}{n}} \text{ and } G_m(n) \approx F^{\inv}\prn{\frac{1-\gamma}{n}}.
    \end{equation}

    \item[(iii)] Finally, we relate the two expressions above (Lemma~\ref{lem:mult-quantile-approx}), by characterizing exactly how two different quantiles of a distribution that follows EVT are related:
    \begin{equation}\label{eq:intro-3}
    F^{\inv}\prn{1 - \frac{c}{n}} \approx c^{-\gamma} F^{\inv}\prn{1 - \frac{1}{n}}, \text{ for maxima and } F^{\inv}\prn{\frac{c}{n}} \approx c^{-\gamma} F^{\inv}\prn{\frac{1}{n}}, \text{ for minima}.
    \end{equation}
\end{enumerate}
The approximate nature of the expressions above holds with equality as $n \to \infty$, a fact that we use to simplify our analysis. To show these simplified expressions, we make use of the \emph{hazard rate} of the distribution; specifically, we rewrite the expected value of the optimal threshold algorithm via the antiderivative of the hazard rate -- for more information on the hazard rate and its formal definition see Section~\ref{sec:preliminaries}. To analyze this new expression, we make heavy use of the theory of \emph{regularly varying functions}, originally developed by Karamata. Regularly varying functions are, roughly speaking, functions that, asymptotically, behave like power functions. We define regularly varying functions in Section~\ref{sec:preliminaries} and provide the basic information that is useful in our proofs -- for more information on the topic of regular variation see \cite{resnick-rv} and \cite{bingham-rv}.

To show (i), we make use of standard facts about the expected value of the extreme value distributions $G_\gamma$ and $G^*_\gamma$, in conjunction with an analysis of the multiplicative scaling factor of $a_n$ from Theorem~\ref{thm:gnedenko} via the lens of regular variation. The main part of our proof, and the most technically involved, is in showing (ii). To do this, we make use of one of the most useful results for regularly-varying functions, namely Karamata's theorem (Lemma~\ref{lem:karamatas}). This theorem provides a simple expression for a complicated integral via the use of EVT; as it turns out, one can rewrite the optimal algorithm's objective in the form of this integral, and then use Karamata's theorem to obtain \eqref{eq:intro-2}. One noteworthy by-product of our analysis is an expression describing the ratio of the expected values of the maximum of $n$ and $m$ samples from the same distribution (Corollary~\ref{cor:mu-ratio-approx-gen}), which may be of independent interest. Finally, (iii) is a straightforward corollary of the fact that $F^{\inv}$ is a regularly varying function.

The choice of domain for $\cD$ in our analysis is done carefully and without loss of generality. Let $\supp(\cD) = [x_*, x^*)$, where $x_* \geq 0$ and $x^* \leq +\infty$. Notice that if $x_* \neq 0$, then in the I.I.D. Min-Prophet Inequality, we can obtain a $(1+\eps)$-competitive ratio, for any $\eps$ and any distribution, by setting a threshold equal to $x_* + \eps$. As $n \to \infty$, the probability that there exists a realization below $x_*+\eps$ goes to $1$. For this reason, we assume that $\supp(\cD) = [0, x^*)$ without loss of generality. In the I.I.D. Max-Prophet Inequality, for distributions with bounded domain, if $x^* < +\infty$, we can obtain a $(1-\eps)$-competitive ratio, for any $\eps$ and any distribution, by setting a threshold equal to $x^*-\eps$. Again, as $n \to \infty$, the probability that there exists a realization above $x^*-\eps$ goes to $1$. This fact will show up in our analysis. Also, notice that the Extreme Value Theorem for minima shows that $\gamma > 0$ implies that the left-most endpoint of the support of $\cD$ has to be $-\infty$. Since these distributions are not studied in the prophet inequality setting, for the I.I.D. Min-Prophet Inequality it must be that $\gamma \leq 0$.

\paragraph{Monotone Hazard Rate (MHR) Distributions.} Distributions with monotonically increasing hazard rate have been extensively studied in the mechanism design literature due to their sought after properties and applications (e.g., see \cite{giannakopoulos-mhr, babaioff-pricing, bhattacharya-pricing, cai-pricing, daskalakis-mhr, revenue-sample-mhr, giannakopoulos-markets, hartline-simple}). These are known as {\em monotone hazard rate (MHR)} (also known as {\em increasing failure rate (IFR)}) distributions. For our third result, we show that for the special case of MHR distributions, $\gamma \leq 0$ for maxima and $\gamma \geq -1$ for minima. Due to the former and Theorem~\ref{thm:unified-acr}, we recover the result of \cite{kesselheim-mhr-ppm} which states that the competitive ratio of the I.I.D. Max-Prophet Inequality goes to $1$ as $n$ goes to infinity and the distribution is MHR. Due to the latter and Theorem~\ref{thm:unified-acr}, we recover the result of \cite{liv-mehta-cpi} which states that the competitive ratio of the I.I.D. Min-Prophet Inequality is at most $2$ as $n$ goes to infinity and the distribution is MHR. Therefore, our analysis unifies and generalizes both prior results.

\begin{theorem}\label{thm:mhr-acr}
Let $\cD$ be an MHR distribution in the domain of attraction of $D_\gamma$. Then, 
for the I.I.D. Max-Prophet Inequality we have $\gamma \leq 0$ and for the I.I.D. Min-Prophet Inequality we have $\gamma \geq -1$.
\end{theorem}

\paragraph{Single-Threshold Algorithms.}
For the rewards maximization setting, the competitive ratio of $\nicefrac{1}{2}$ in the classical prophet inequality is achievable through simple single-threshold algorithms~\cite{sam-cahn, klein-wein} of the form ``accept the first $X_i \geq \tau$ for some threshold $\tau$'', and is known to be tight. Furthermore, there exist simple online algorithms that achieve constant-factor approximations even for general multi-dimensional settings with complicated constraints (e.g. matroids\footnote{A matroid is a non-empty downward-closed feasibility constraint where if $A, B \subseteq N$ are both feasible and $|A| < |B|$, there exists an element $e \in B \setminus A$ such that $A \cup \{e\}$ is also feasible. The feasible sets of a matroid are called \emph{independent sets}.}, matchings, and more) \cite{klein-wein, Alaei14, willma-adversarial, gravin-bipartite, ezra}. For the cost minimization setting, \cite{liv-mehta-cpi} shows a qualitatively different result. Specifically, the authors obtain tight competitive ratios, for the special case of entire distributions, that are poly-logarithmic in $n$.

We generalize the results of \cite{liv-mehta-cpi} and obtain tight poly-logarithmic competitive ratios via single-threshold algorithms for distributions $F \in D_\gamma$ for some $\gamma$. We note that entire distributions form a special case as every entire distribution is in the domain of attraction of some $\gamma$. As a corollary of our analysis, we show that the class of distributions for which single-threshold algorithms achieve a constant competitive ratio in the I.I.D. Min-Prophet Inequality is exactly the distributions $F \in D_0$.

\begin{theorem}\label{thm:single-threshold-evt}
For the I.I.D. Min-Prophet Inequality and any distribution $F \in D_\gamma$, where $\gamma < 0$, there exists a single-threshold algorithm that achieves a competitive ratio of $\bigO{\prn{\log{n}}^{-\gamma}}$. Furthermore, this factor is tight; no single-threshold algorithm can achieve a $\ltlo{\prn{\log{n}}^{-\gamma}}$ competitive ratio for all instances.
\end{theorem}

\paragraph{Multi-Unit Prophet Inequality.} Since the guarantees of the above theorem are poly-logarithmic in $n$, a natural question is whether single-threshold algorithms can achieve constant competitive ratios in any minimization setting. We answer this positively, by considering the $k$-multi-unit prophet inequality, in which we are presented with $n$ random variables drawn I.I.D. from a known distribution following EVT and have to select \emph{at least} $k$ values, with the goal of minimizing the sum of selected values. A single-threshold algorithm for such a setting selects the first $k$ realizations below a threshold $T$. If the number $r$ of such realizations is smaller than $k$, it then is forced to select the last $k-r$ realizations independent of their value, to guarantee feasibility. We show that, when $k \geq \log{n}$, the competitive ratio of single-threshold algorithms undergoes a phase transition and, from poly-logarithmic in $n$, becomes constant.

\begin{theorem}\label{thm:single-threshold-evt-multi}
For every instance of the $k$-Multi-Unit I.I.D. Min-Prophet Inequality with distribution $F \in D_\gamma$, if $k \geq \log{n}$, there exists a single-threshold algorithm that achieves a constant competitive ratio.
\end{theorem}

\subsection{Related Work}\label{sec:related-work}

Our work is most closely related to the long line of work that considers the case of I.I.D. random variables drawn from a known distribution, which dates back to Hill and Kertz \cite{hill-kertz}. As stated previously, Kertz \cite{kertz} showed that the competitive ratio in the I.I.D. case approaches $\approx 0.745$ as $n$ goes to infinity and conjectured its tightness. A simpler proof of this can be found in \cite{saint-mont}. The bound of $\approx 0.745$ was shown to be tight by Correa, Foncea, Hoeksma, Oosterwijk and Vredeveld \cite{correa-iid}. The proofs of both the upper and lower bounds were recently simplified, by \cite{better-tightness} and \cite{renato-iid} respectively. Extreme Value Theory was first used to analyze prophet inequalities by Kennedy and Kertz \cite{kennedy-kertz}, who gave the closed form of the asymptotic competitive ratio for the maximization setting. Later, Correa, Pizarro and Verdugo \cite{jose-dana-victor-evt} analyzes single-threshold algorithms for the I.I.D. Max-Prophet Inequality. For a comprehensive review of the I.I.D. setting see the survey by Correa, Foncea, Hoeksma, Oosterwijk and Vredeveld \cite{correa-survey}.

One can consider a more general setting in which the distributions can differ but the arrival order of the random variables is chosen uniformly at random. This setting, called the \emph{prophet secretary}, was introduced by Esfandiari, Hajiaghayi, Liaghat and Monemizadeh \cite{esf-prophsec}, where they gave an algorithm that achieves a $\prn{1 - \f{1}{e}}$-competitive ratio and showed that no algorithm can achieve a factor better than $0.75$. Ehsani, Hajiaghayi, Kesselheim and Singla \cite{EhsaniHKS18} extend this result to matroid constraints. The factor of $1 - \f{1}{e}$ was recently beaten \cite{azar, correa-random}. The best currently known ratio is $0.672$, obtained by Harb \cite{farouk-poisson-random}. The upper bound was also improved, first by \cite{correa-random} to $0.732$ and recently by \cite{bubna-chiplunkar} to $0.7254$. The best known upper bound stands at $0.7235$ by \cite{raimundo-random}.

Another variant of particular interest is the case where the algorithm is allowed to choose the arrival order, called the \emph{free order} or \emph{order selection} setting. Settling the competitive ratio of this setting is a significant open problem in the area of prophet inequalities. Early progress was made by \cite{beyhaghi-free-order,peng-tang}. Currently, the best known factor is $0.7258$, obtained by \cite{bubna-chiplunkar}. It is still open whether one can achieve the tight $\approx 0.745$ factor of the I.I.D. prophet inequality even in the free order setting.

The $\f{1}{2}$-competitive factor guaranteed by the classical prophet inequality for adversarial arrival order has been shown to hold for more general classes of downwards-closed constraints, all the way up to matroids \cite{klein-wein}. For the special case of $k$-uniform matroids, where one can select up to $k$ values, Alaei \cite{Alaei14} showed a $\prn{1 - \frac{1}{\sqrt{k+3}}}$-competitive ratio. This was recently improved for small $k$ via the use of a static threshold \cite{suchi-small-k} and later made tight for all $k$ \cite{willma-adversarial}. For the same setting and a random arrival order, Arnosti and Ma \cite{willma-random} recently gave a surprising and quite beautiful single-threshold algorithm that achieves the best competitive ratio of $1 - e^{-k} \frac{k^k}{k!}$. Ezra, Feldman, Gravin and Tang \cite{ezra} showed a $0.337$-prophet inequality for matching constraints. Rubinstein \cite{rubin} considered general downwards-closed feasibility constraints and obtained logarithmic approximations. The standard setting has been extended to combinatorial valuation functions \cite{rs,chek-liv}, where one can obtain a constant competitive ratio when maximizing a submodular\footnote{A set function $f : 2^N \to \R$ is called \emph{submodular} if for any sets $A, B \subseteq N$ we have $f(A) + f(B) \geq f(A \cup B) + f(A \cap B)$.} function but a logarithmic hardness is known for subadditive\footnote{A set function $f : 2^N \to \R$ is called \emph{subadditive} if for any sets $A, B \subseteq N$ we have $f(A) + f(B) \geq f(A \cup B)$.} functions \cite{rs}. Recently, \cite{suchi-non-adaptive-graphic-matroid} studied non-adaptive threshold algorithms for matroid constraints and gave the first constant-factor competitive algorithm for graphic matroids. More general feasibility constraints have also been studied in the random arrival order case, i.e. for matroids \cite{aw-random} and matchings \cite{brubach, tristan}. Qin, Rajagopal, Vardi and Wierman \cite{convex-pi,convex-pi-2} study the related problem of \emph{convex prophet inequalities}, in which, instead of costs, one sequentially observes random \emph{cost functions} and needs to assign a total mass of $1$ across all functions in an online manner. The model of \cite{liv-mehta-cpi} is a special case of their model, for constant cost functions and identical distributions.

\paragraph{Applications.} Prophet inequalities have been successfully applied to provide guarantees on the social welfare and revenue of simple yet approximately optimal auctions. This was in part motivated by the intractability of optimal mechanisms for selling items \cite{pricing-lotteries, hart-nisan-menu, daskalakis-intractability-1, daskalakis-intractability-2}. Hajiaghayi, Kleinberg and Sandholm \cite{haji}, and later Chawla, Hartline, Malec and Sivan \cite{ChawlaHMS}, were the pioneers in the use of prophet inequalities to analyze (sequential) posted price mechanisms for selling items. Their results led to a significant effort to understand how the expected revenue of an optimal posted price mechanism compares to that of the optimal auction \cite{chawla07, qiqi, BlumHolen08, Adamczyk17, Alaei14, feldman-combinatorial, Dutting2, kesselheim-mhr-ppm, Dobzinski-Comb-Auc, Dobzinski-Comb-Auct-Impos, AssadiKS, assad-singla, dutting-combinatorial}. This series of work finally culminated in a beautiful result by Correa and Cristi \cite{correa-cristi}, where they gave the first constant-factor guarantees for combinatorial settings and subadditive valuations.

In a surprising result, Correa, Foncea, Pizarro and Verdugo \cite{CorreaPricing} showed that the reverse direction also holds, establishing an equivalence between finding stopping rules in an optimal stopping problem and designing optimal posted price mechanisms -- for more information on these applications see a survey by Lucier \cite{lucier-survey}. Recently, \cite{trading-proph} initiated the study of buy-and-sell prophet inequalities, named {\em trading prophets}, where they show constant factor guarantees even with single-threshold algorithms.

Cost minimization is equally relevant in practice as the rewards maximization setting -- there are several scenarios in which a procuring agency needs to procure a good or service to guarantee some fixed supply while minimizing cost. In electricity markets, for example, entities that produce power enter long-term contracts that obligate them to serve consumers, regardless of their own power generation capacity. Therefore, they must estimate how much power they can produce themselves, and cover the difference in demand by procuring power from third-party producers \cite{electricity-markets-1,electricity-markets-2}. Other examples include natural provisioning resources in cloud computing \cite{cloud-computing-application} and gas supply markets \cite{gas-supply-market}.

\paragraph{Organization.} Section~\ref{sec:preliminaries} contains all relevant definitions and technical background needed for our results. Section~\ref{sec:unified} characterizes the optimal (multiple) threshold algorithm and contains our unified analysis for both the I.I.D. Max-Prophet Inequality and the I.I.D. Min-Prophet Inequality. In Section~\ref{sec:single-threshold}, we show our results for single-threshold algorithms. Finally, we offer some concluding remarks in Section~\ref{sec:conclusion}.

\section{Preliminaries}\label{sec:preliminaries}

\paragraph{Models}
Let $X_1, \dots, X_n$ denote random variables drawn independently from a known distribution $\cD$ supported on $[0, x^*)$, where $x^* \leq +\infty$. In the prophet inequality setting, we are presented with the realizations of $X_i$ sequentially and at each step $i$ we must make an immediate and irrevocable decision to accept or reject $X_i$. The process ends when we accept a realization and, once rejected, a realization cannot be obtained in the future. The benchmark is an all-knowing prophet who can see all realizations in advance and always select the optimal. In the I.I.D. Max-Prophet Inequality, the goal is to maximize the selected realization and the prophet's objective is $\E[\max_i X_i]$, whereas in the I.I.D. Min-Prophet Inequality, the goal is to minimize the selected realization and the prophet's objective is $\E[\max_i X_i]$ (but we are forced to select at least one realization). We also consider the $k$-Multi-Unit I.I.D. Min-Prophet Inequality, in which we are forced to select \emph{at least} $k$ realizations, with the goal of minimizing the sum of selected realizations. The prophet's objective now becomes the expected value of the $k$ smallest random variables.

\paragraph{Order Statistics}
Suppose we reorder the random variables such that $X_{(1)} \leq \dots \leq X_{(n)}$. Then, $X_{(i)}$ is called the $i$-th \emph{order statistic} of $\cD$. We denote the expected value of the $i$-th order statistic of $n$ samples from $\cD$ by $\mu_{i:n} = \E_{\cD}[X_{(i)}]$. Of special interest to us are the expectation of the \emph{last} and \emph{first} order statistic, i.e. the largest $\mu_{n:n}$ and smallest $\mu_{1:n}$ expected values, since they capture the prophet's objective in the I.I.D. Max-Prophet Inequality and I.I.D. Min-Prophet Inequality, respectively.

\paragraph{Left Continuous Inverse}
Let $F : [0, +\infty) \to [0, 1]$, where $F(x) = \Pr_{X \sim \cD}\brk{X \leq x}$, and $f : [0, +\infty) \rightarrow [0,1]$ denote the \emph{Cumulative Distribution Function (CDF)} and \emph{Probability Density Function (PDF)} of $\cD$, respectively.
\begin{definition}[Left Continuous Inverse]\label{def:ch2-inverse}
Let $F$ be a non-decreasing function on $\R$. The \emph{(left continuous) inverse} of $F$ is defined as
\[
F^{\inv}(y) = \inf\set{x \midd F(x) \geq y}.
\]
\end{definition}
The above definition works with the convention that the infimum of an empty set is $+\infty$. For more information on the left continuous inverse see \cite{resnick-rv} (Section $0.2$). In particular, if $F$ denotes the CDF of a distribution $\cD$, then $F^{\inv}(y)$ denotes the quantile function of $\cD$, i.e. $F^{\inv}(y)$ is the smallest value $\tau$ for which $\Pr_{X \sim \cD}[X \leq \tau] \geq y$.

\paragraph{Asymptotic Competitive Ratio}
We use $G_M(n)$ and $G_m(n)$ to denote the expected value of the optimal threshold policy for the I.I.D. Max-Prophet Inequality and the I.I.D. Min-Prophet Inequality settings with $n$ random variables, respectively. To avoid confusion, we denote the asymptotic competitive ratio in the I.I.D. Max-Prophet Inequality by $\lambda_M$ and in the I.I.D. Min-Prophet Inequality by $\lambda_m$.
\begin{definition}[Asymptotic Competitive Ratio]
Consider an instance of the I.I.D. Max-Prophet Inequality (resp. I.I.D. Min-Prophet Inequality) with $n$ random variables. Then, the \emph{asymptotic competitive ratio} $\lambda$ is
\begin{align*}
\lambda_M &\triangleq \lim_{n \to \infty} {\frac{G_M(n)}{\mu_{n:n}}}, \qquad \text{in the I.I.D. Max-Prophet Inequality}, \\
\lambda_m &\triangleq \lim_{n \to \infty} {\frac{G_m(n)}{\mu_{1:n}}}, \qquad \text{in the I.I.D. Min-Prophet Inequality}.
\end{align*}
\end{definition}

Since we analyze the asymptotic competitive ratio as $n \to \infty$, we use $\approx$ to denote the asymptotic equality of two functions of $n$.
\begin{definition}[Asymptotic Equality]
For every pair of functions $a, b : \N \to \Rp$, we use $a(n) \approx b(n)$ to denote the following:
\[
\lim_{n \to \infty} {\frac{a(n)}{b(n)}} = 1.
\]
\end{definition}

\paragraph{Regularly Varying Functions}
Our analysis relies heavily on the theory of \emph{regularly varying functions}, originally developed by Karamata. Regularly varying functions are, roughly speaking, functions that behave asymptotically like power functions. We present some basic definitions here; for more information on the topic see \cite{resnick-rv} and \cite{bingham-rv}.
\begin{definition}\label{def:reg-var}
Let $f : \Rp \to \R$ be a Lebesgue measurable function that is eventually positive. We say that $f$ is \emph{regularly varying} (at infinity) if, for some $\alpha \in \R$ and every $x > 0$
\[
\lim_{t \to \infty} {\frac{f(tx)}{f(t)}} = x^\alpha.
\]
In this case, we indicate this as $f \in RV_\alpha$.
\end{definition}

In the above definition, $\alpha$ is called the \emph{index} of regular variation and whenever $\alpha = 0$, we say that $f$ is \emph{slowly varying}. Furthermore, we say that $f(x)$ is regularly varying at $0$ if and only if $f(\f{1}{x})$ is regularly varying at infinity. If $f \in RV_\alpha$, then $L(x) = \f{f(x)}{x^\alpha} \in RV_0$. In general, one can represent any $f \in RV_\alpha$ as $f(x) = x^\alpha L(x)$, where $L$ is a slowly-varying function.

\paragraph{Gamma Function}
The Gamma ($\Gamma$) function -- which is an extension of the factorial function over the reals -- and its relatives arise in our closed form of the ACR. For $x > 0$, it is defined as $\Gamma(x) = \int^\infty_0 {t^{x-1} e^{-t} \dif t}$. In particular, $\Gamma(n+1) = n!$ for every $n \in \N$. Like the factorial function, the Gamma function also satisfies the following recurrence $\Gamma(x+1) = x \Gamma(x)$. For a more extensive treatment along with many folklore results about the function, see \cite{gamma-book}.

\paragraph{Monotone Hazard Rate (MHR) Distributions}
Intuitively, for discrete distributions, the hazard rate for maxima (resp. minima) at a point $t$ represents the probability that an event occurs at time $t$, given that the event has not occurred (resp. has occured) up to time $t$. For continuous distributions, the hazard rate instead quantifies the instantaneous rate of the event's occurrence at time $t$.
\begin{definition}[Hazard Rate]\label{def:hazard-rate}
For a distribution $\cD$ with cumulative distribution function $F$ and probability density function $f$, the {\em maxima hazard rate} $h$ and {\em minima hazard rate} $r$ of $\cD$ are defined as
\[
h(x) \triangleq \frac{f(x)}{1 - F(x)} \qquad \text{and} \qquad r(x) \triangleq \frac{f(x)}{F(x)},
\]
for all $x$ in the support of $\cD$.

Furthermore, let $H$ and $R$ denote the antiderivatives of $h$ and $r$, which we call the {\em cumulative hazard rates} of $\cD$ for maxima and minima respectively:
\[
H(x) \triangleq \int^x_0 {h(u) \dif u} \qquad \text{and} \qquad R(x) \triangleq \int^x_0 {r(u) \dif u}.
\]
\end{definition}
Notice that,
\[
H(x) = \int^x_0 {h(u) \dif u} = \int^x_0 {\frac{f(u)}{1 - F(u)} \dif u} = - \int^x_0 {\prn{\ln\prn{1 - F(u)}}' \dif u} = - \ln\prn{1 - F(x)},
\]
which implies that $1 - F(x) = e^{-H(x)}$. Similarly,
\[
R(x) = \int^x_0 {r(u) \dif u} = \int^x_0 {\frac{f(u)}{F(u)} \dif u} = \int^x_0 {\prn{\ln\prn{F(u)}}' \dif u} = \ln\prn{F(x)},
\]
which implies that $F(x) = e^{R(x)}$.

Distributions with monotonically increasing hazard rate (for maxima) have found a special place within mechanism design literature, originally introduced for the study of revenue maximization. They are known as MHR (or IFR for increasing failure rate) distributions.

\begin{definition}[Monotone Hazard Rate Distribution]\label{def:mhr}
A distribution $\cD$ is called a {\em Monotone Hazard Rate (MHR)} distribution if and only if the hazard rate function $h$ of $\cD$ is non-decreasing.
\end{definition}

\section{A Unified Approach}\label{sec:unified}

In this section, we prove our main theorem.

\begin{reptheorem}{thm:unified-acr}
Let $F \in D_\gamma$. Then,
\begin{itemize}
    \item for the I.I.D. Max-Prophet Inequality and every $\gamma \in \R$, the asymptotic competitive ratio of the optimal threshold policy as $n \to \infty$ is
    \[
    \lambda_M= \min\set{\frac{\prn{1 - \gamma}^{-\gamma}}{\Gamma\prn{1 - \gamma}}, 1}.
    \]
    \item for the I.I.D. Min-Prophet Inequality, we have that $\gamma \leq 0$, and the asymptotic competitive ratio of the optimal threshold policy as $n \to \infty$ is
    \[
    \lambda_m = \max\set{\frac{\prn{1 - \gamma}^{-\gamma}}{\Gamma\prn{1 - \gamma}}, 1}.
    \]
\end{itemize}
\end{reptheorem}

We split our analysis into the following cases and show that
\begin{itemize}
\item For I.I.D. Max-Prophet Inequality,
\begin{itemize}
    \item $\lambda_M = 1$, for $\gamma \geq  1$,
    \item $\lambda_M = \frac{\prn{1 - \gamma}^{-\gamma}}{\Gamma(1 - \gamma)}$, for $\gamma \in (0,1)$,
    \item $\lambda_M = 1$, for $\gamma = 0$,
    \item $\lambda_M = 1$, for $\gamma < 0$.
\end{itemize}
\item For I.I.D. Min-Prophet Inequality,
\begin{itemize}
    \item $\lambda_m = 1$, for $\gamma = 0$,
    \item $\lambda_m = \frac{\prn{1 - \gamma}^{-\gamma}}{\Gamma(1 - \gamma)}$, for $\gamma < 0$,
\end{itemize}
\end{itemize}
For the I.I.D. Max-Prophet Inequality setting, the case of $\gamma \geq 1$ is easy to see as it implies that $\E[X] = + \infty$, and thus accepting the first random variable trivially yields a competitive ratio of $1$. For simplicity and ease of presentation, we present the proof for $\gamma \in (0,1)$ and I.I.D. Max-Prophet Inequality and $\gamma < 0$ and I.I.D. Min-Prophet Inequality. In the remaining cases, the proof is very similar, except for minor details, and for this reason the proofs of the other cases can be found in Appendix~\ref{app:proofs}. Recall that the case $\gamma > 0$ for I.I.D. Min-Prophet Inequality is impossible as it would imply that the left-most endpoint of the domain is $-\infty$.

Recall that $G_M(n)$ and $G_m(n)$ denote the expected value of the optimal threshold policy for the I.I.D. Max-Prophet Inequality and the I.I.D. Min-Prophet Inequality settings, respectively. Also, let $\lambda_M(n)$ and $\lambda_m(n)$ denote the competitive ratios for the I.I.D. Max-Prophet Inequality and the I.I.D. Min-Prophet Inequality settings, for every $n \geq 1$. Moreover, assume that $\cD = [0, x^*)$, where $x^* \leq +\infty$.

\begin{theorem}\label{thm:unified-acr-case-1}
Let $F \in D_\gamma$. Then,
\begin{itemize}
    \item for the I.I.D. Max-Prophet Inequality, if $\gamma \in (0,1)$, the asymptotic competitive ratio of the optimal threshold policy as $n \to \infty$ is
    \[
    \lambda_M = \frac{\prn{1 - \gamma}^{-\gamma}}{\Gamma\prn{1 - \gamma}}.
    \]
    \item for the I.I.D. Min-Prophet Inequality, if $\gamma < 0$, the asymptotic competitive ratio of the optimal threshold policy as $n \to \infty$ is
    \[
    \lambda_M = \frac{\prn{1 - \gamma}^{-\gamma}}{\Gamma\prn{1 - \gamma}}.
    \]
\end{itemize}
\end{theorem}

The proofs of the following theorems can be found in Appendix~\ref{app:proofs}.

\begin{theorem}\label{thm:unified-acr-case-2}
Let $F \in D_0$. Then, for both the I.I.D. Max-Prophet Inequality and the I.I.D. Min-Prophet Inequality, the asymptotic competitive ratio of the optimal threshold policy as $n \to \infty$ is $1$.
\end{theorem}

\begin{theorem}\label{thm:unified-acr-case-3}
Let $F \in D_\gamma$, for $\gamma < 0$. Then, for the I.I.D. Max-Prophet Inequality, the asymptotic competitive ratio of the optimal threshold policy as $n \to \infty$ is $1$.
\end{theorem}

The proof of Theorem~\ref{thm:unified-acr} now follows from Theorems~\ref{thm:unified-acr-case-1},~\ref{thm:unified-acr-case-2} and~\ref{thm:unified-acr-case-3}.

The remainder of this section is dedicated in showing Theorem~\ref{thm:unified-acr-case-1}. We start by getting a recursive form of $G_M(n)$ and $G_m(n)$ which will be more useful, in terms of the distribution's inverse hazard rate. The next lemma is folklore, but we present its proof in Appendix~\ref{app:proofs} for completeness.

\begin{lemma}\label{lem:opt-dp-g}
For any $n > 1$, we have
\[
G_M(n) = G_M(n-1) + \int^{x^*}_{G_M(n-1)} {\prn{1 - F(u)} \dif u},
\]
and
\[
G_m(n) = \int^{G_m(n-1)}_0 {\prn{1 - F(u)} \dif u}.
\]
\end{lemma}

Recall the cumulative hazard rate functions $H(x) = -\log\prn{1 - F(x)}$ and $R(x) = \log{F(x)}$.

\begin{lemma}\label{lem:opt-dp-hazard}
For any $n > 1$, we have
\[
G_M(n) = G_M(n-1) \prn{1 - e^{-H(G_M(n-1))}} + \int^{\exp\prn{-H(G_M(n-1))}}_0 {H^{\inv}\prn{-\log{u}} \dif u},
\]
and
\[
G_m(n) = G_m(n-1) \prn{1 - e^{R(G_m(n-1))}} + \int^{\exp\prn{R(G_m(n-1))}}_0 {R^{\inv}\prn{\log{u}} \dif u}.
\]
\end{lemma}
\begin{proof}
Substituting the definitions of $H$ and $R$ into Lemma~\ref{lem:opt-dp-g}, we have
\[
G_M(n) = G_M(n-1) + \int^{x^*}_{G_M(n-1)} {e^{-H(u)} \dif u},
\]
and
\[
G_m(n) = \int^{G_m(n-1)}_0 {1 - e^{R(u)} \dif u}.
\]
Let $x = e^{-H(u)}, t = e^{R(u)}$, which implies that $\dif u = \prn{H^{\inv}\prn{-\log{x}}}^{'} \dif x = \prn{R^{\inv}\prn{\log{t}}}^{'} \dif t$. Thus,
\[
G_M(n) = G_M(n-1) + \int^0_{\exp\prn{-H(G_M(n-1))}} {x \cdot \prn{H^{\inv}\prn{-\log{x}}}^{'} \dif x},
\]
and
\[
G_m(n) = \int^{\exp\prn{R(G_m(n-1))}}_0 {(1 - t) \prn{R^{\inv}\prn{\log{t}}}^{'} \dif t}.
\]
Integrating by parts, we obtain
\begin{align*}
G_M(n) &= G_M(n-1) + \brk{x \cdot H^{\inv}\prn{-\log{x}} }^0_{\exp\prn{-H(G_M(n-1))}} - \int^0_{\exp\prn{-H(G_M(n-1))}} {H^{\inv}\prn{-\log{x}} \dif x} \\
&= G_M(n-1) - G_M(n-1) \cdot e^{-H(G_M(n-1))} + \int^{\exp\prn{-H(G_M(n-1))}}_0 {H^{\inv}\prn{-\log{x}} \dif x} \\
&= G_M(n-1) \prn{1 - e^{-H(G_M(n-1))}} + \int^{\exp\prn{-H(G_M(n-1))}}_0 {H^{\inv}\prn{-\log{x}} \dif x},
\end{align*}
and
\begin{align*}
G_m(n) &= \brk{(1 - t) R^{\inv}\prn{\log{t}} }^{\exp\prn{R(G_m(n-1))}}_0  + \int^{\exp\prn{R(G_m(n-1))}}_0 {R^{\inv}\prn{\log{t}} \dif t} \\
&= G_m(n-1) \prn{1 - e^{R(G_m(n-1))}}  + \int^{\exp\prn{R(G_m(n-1))}}_0 {R^{\inv}\prn{\log{t}} \dif t}.
\end{align*}
\end{proof}

Next, we need to analyze the integrals of the inverse hazard rates. To do this, we need the following property of $H$ and $R$ from the theory of regular variation of functions.

\begin{lemma}\label{lem:rv-hazard}
Let $U_M(x) = H^{\inv}(-\log{x})$ and $U_m(x) = R^{\inv}(\log{x})$. If $F \in D_\gamma$, then $U_m \in RV_{-\gamma}$ at $0$. Furthermore, for $\gamma \in [0,1)$, we have $U_M \in RV_{-\gamma}$ at $0$, and for $\gamma \leq 0$, we have $x^* < +\infty$ and $x^* - U_M \in R_{-\gamma}$ at $0$.
\end{lemma}
\begin{proof}
Since $1 - F(x) = e^{-H(x)}$ and $F(x) = e^{R(x)}$, we have that $U_M(x) = F^{\inv}\prn{1 - x}$ and $U_m(x) = F^{\inv}\prn{x}$.

Let $\gamma \in [0,1)$. Since $F \in D_\gamma$, by \cite{dehaan-ferreira-evt} (Corollary 1.2.10), we have that $U_M\prn{\f{1}{x}} \in RV_\gamma$ which implies that
\[
\lim_{t \to \infty} {\frac{U_M\prn{\f{1}{tx}}}{U_M\prn{\f{1}{t}}}} = x^\gamma.
\]
Therefore,
\[
\lim_{t \to 0^+} {\frac{U_M(tx)}{U_M(t)}} = \lim_{t \to 0^+} {\frac{U_m(tx)}{U_m(t)}} = x^{-\gamma},
\]
and $U_M, U_m \in RV_{-\gamma}$ at $0$.

Next, let $\gamma < 0$.  Since $F \in D_\gamma$, by \cite{dehaan-ferreira-evt} (Corollary 1.2.10), we have that $x^* - U_M\prn{\f{1}{x}}$, $U_m\prn{\f{1}{x}} \in RV_\gamma$, which implies that
\[
\lim_{t \to \infty} {\frac{x^* - U_M\prn{\f{1}{tx}}}{x^* - U_M\prn{\f{1}{t}}}} = \lim_{t \to \infty} {\frac{U_m\prn{\f{1}{tx}}}{U_m\prn{\f{1}{t}}}} = x^\gamma.
\]
Therefore,
\[
\lim_{t \to 0^+} {x^* - \frac{U_M(tx)}{x^* - U_M(t)}} = \lim_{t \to 0^+} {\frac{U_m(tx)}{U_m(t)}} = x^{-\gamma},
\]
and $x^* - U_M, U_m \in RV_{-\gamma}$ at $0$.
\end{proof}

Next, we present a famous theorem in the theory of regularly varying functions, due to Karamata, that we use to analyze the asymptotic behaviour of the integral of Lemma~\ref{lem:opt-dp-hazard}. For more details on this see \cite{dehaan-ferreira-evt} (Appendix B) and \cite{resnick-rv}.

\begin{lemma}[Karamata's Theorem]\label{lem:karamatas}
For $F \in D_\gamma$ and large enough $n$, we have
\[
\int^{\exp\prn{-H(G_M(n-1))}}_0 {H^{\inv}\prn{-\log{u}} \dif u} \approx \frac{G_M(n-1) \cdot e^{-H(G_M(n-1))}}{1 - \gamma}, \quad \text{for } \gamma \leq 1,
\]
and
\[
\int^{\exp\prn{R(G_m(n-1))}}_0 {R^{\inv}\prn{\log{u}} \dif u} \approx \frac{G_m(n-1) \cdot e^{R(G_m(n-1))}}{1 - \gamma}, \quad \text{for } \gamma \leq 0.
\]
\end{lemma}
\begin{proof}
Notice that, as $n \to \infty$, $e^{-H(G_M(n-1))}, e^{R(G_m(n-1))} \to 0$. The result then follows directly from Karamata's theorem (\cite{dehaan-ferreira-evt}, Theorem B.1.5), since $H^{\inv}(-\log{x}), R^{\inv}(\log{x}) \in RV_{-\gamma}$ at $0$, by Lemma~\ref{lem:rv-hazard}.
\end{proof}

Now we can combine Lemmas~\ref{lem:opt-dp-hazard} and~\ref{lem:karamatas} to obtain a simplified approximation to $G_M(n)$ and $G_m(n)$ for large $n$.
\begin{lemma}\label{lem:gn-done}
For every $F \in D_\gamma$, where $\gamma < 1$ for the I.I.D. Max-Prophet Inequality, and large enough $n$, we have
\begin{equation}\label{eq:gn-done-1}
G_M(n) \approx G_M(n-1) \prn{1 - e^{-H(G_M(n-1))}\prn{1 - \frac{1}{1-\gamma}}},
\end{equation}
and
\begin{equation}\label{eq:gn-done-2}
G_m(n) \approx G_m(n-1) \prn{1 - e^{R(G_m(n-1))}\prn{1 - \frac{1}{1-\gamma}}}.
\end{equation}
\end{lemma}

We now turn our attention to the expectation of the first and last order statistic of $F$, i.e. to $\mu_{1:n}$ and $\mu_{n:n}$.

\begin{lemma}\label{lem:mu-approx}
For $F \in D_\gamma$ and large enough $n$, we have
\[
\mu_{n:n} \approx \begin{cases}
\Gamma(1 - \gamma) F^{\inv}\prn{1 - \frac{1}{n}}, & \text{for } \gamma \in (0,1), \\
F^{\inv}\prn{1 - \frac{e^{-\gamma^*}}{n}}, & \text{for } \gamma = 0, \\
x^* - \Gamma(1 - \gamma) \prn{x^* - F^{\inv}\prn{1 - \frac{1}{n}}}, & \text{for } \gamma < 0,
\end{cases}
\]
and
\[
\mu_{1:n} \approx \begin{cases}
\Gamma(1 - \gamma) F^{\inv}\prn{\frac{1}{n}}, & \text{for } \gamma < 0, \\
F^{\inv}\prn{\frac{e^{-\gamma^*}}{n}}, & \text{for } \gamma = 0.
\end{cases}
\]
where $\gamma^* \approx 0.577$ is the Euler-Mascheroni constant.
\end{lemma}
\begin{proof}
Since $F \in D_\gamma$, by Theorem~\ref{thm:gnedenko}, we know that there exist sequences $a_n > 0, b_n \in \R$ such that
\begin{equation}\label{eq:gnedenko-mu-approx-max}
\lim_{n \to \infty} F^n(a_n x + b_n) = G_\gamma(x).
\end{equation}
Let $M_n = \max\set{X_1, \dots, X_n}$. If \eqref{eq:gnedenko-mu-approx-max} is satisfied, it has to be satisfied for
\begin{itemize}
    \item $a_n = U_M\prn{\f{1}{n}}$ and $b_n = 0$, if $\gamma \in (0,1)$,
    \item $b_n = U_M\prn{\f{1}{n}}$ and appropriately chosen $a_n$ if $\gamma = 0$, and
    \item $a_n = x^* - U_M\prn{\f{1}{n}}$ and $b_n = x^*$, if $\gamma < 0$,
\end{itemize}
by \cite{dehaan-ferreira-evt} (Corollary 1.2.4). Let $Y_n = \max\set{\f{X_1 - b_n}{a_n}, \dots, \f{X_n - b_n}{a_n}} = \frac{M_n - b_n}{a_n}$. The above imply that $\lim_{n \to \infty} Y_n$ converges in distribution to a random variable $Z$ distributed according to $G_\gamma$. Notice that, for $0 < \gamma < 1$, we have $\E[Z] = \Gamma(1 - \gamma)$, for $\gamma = 0$, we have $\E[Z] = \gamma^*$ and for $\gamma < 0$, we have $\E[Z] = - \Gamma(1 - \gamma)$ (\cite{dehaan-ferreira-evt}, Theorem 5.3.1). This implies that, for large enough $n$ and $\gamma \in (0,1)$
\[
\E[Y_n] \approx \Gamma(1 - \gamma) \iff \mu_{n:n} \approx \Gamma(1 - \gamma) a_n = \Gamma(1 - \gamma) F^{\inv}\prn{1 - \frac{1}{n}}.
\]
For $\gamma = 0$, we have
\begin{equation}\label{eq:mu-approx-max-1}
\E[Y_n] \approx \gamma^* \iff \mu_{n:n} \approx a_n \gamma^* + b_n = a_n \gamma^* + U_M\prn{\f{1}{n}} = a_n \gamma^* + F^{\inv}\prn{1 - \frac{1}{n}}.
\end{equation}
Finally, for $\gamma = 0$, by \cite{dehaan-ferreira-evt} (Theorem 1.1.6), we have that, for any $x > 1$
\begin{equation}\label{eq:mu-approx-max-2}
a_n \approx \frac{U_M\prn{\f{1}{x n}} - U_M\prn{\f{1}{n}}}{\log{x}} = \frac{F^{\inv}\prn{1 - \frac{1}{x n}} - F^{\inv}\prn{1 - \frac{1}{n}}}{\log{x}}.
\end{equation}
Combining \eqref{eq:mu-approx-max-1} and \eqref{eq:mu-approx-max-2}, we get
\[
\mu_{n:n} \approx \frac{\gamma^*}{\log{x}} F^{\inv}\prn{1 - \frac{1}{x n}} + \prn{1 - \frac{\gamma^*}{\log{x}}} F^{\inv}\prn{1 - \frac{1}{n}}.
\]
Setting $x = e^{\gamma^*}$ yields
\[
\mu_{n:n} \approx F^{\inv}\prn{1 - \frac{e^{-\gamma^*}}{n}}.
\]
Finally, for $\gamma < 0$, we have
\[
\E[Y_n] \approx -\Gamma(1 - \gamma) \iff \mu_{n:n} \approx -\Gamma(1 - \gamma) a_n  + b_n = x^* - \Gamma(1 - \gamma) \prn{x^* - F^{\inv}\prn{1 - \frac{1}{n}}}.
\]

Similarly, by Theorem~\ref{thm:gnedenko}, we know that there exist sequences $a'_n > 0, b'_n \in \R$ such that
\begin{equation}\label{eq:gnedenko-mu-approx-min}
\lim_{n \to \infty} \prn{1 - F(a_n x + b_n)}^n = G^*_\gamma(x).
\end{equation}
Let $m_n = \min\set{X_1, \dots, X_n}$. If \eqref{eq:gnedenko-mu-approx-min} is satisfied, it has to be satisfied for $a'_n = U_m\prn{\f{1}{n}}$ and $b'_n = 0$, if $\gamma < 0$ and for $b'_n = U_m\prn{\f{1}{n}}$ and appropriately chosen $a'_n$ if $\gamma = 0$ (\cite{dehaan-ferreira-evt}, Corollary 1.2.4). Let $Y'_n = \min\set{\f{X_1 - b'_n}{a'_n}, \dots, \f{X_n - b'_n}{a'_n}} = \frac{m_n - b'_n}{a'_n}$. The above imply that $\lim_{n \to \infty} Y'_n$ converges in distribution to a random variable $Z'$ distributed according to $G^*_\gamma$. Notice that, for $\gamma < 0$, we have $\E[Z'] = \Gamma(1 - \gamma)$, whereas for $\gamma = 0$, we have $\E[Z'] = -\gamma^*$. This implies that, for large enough $n$,
\[
\E[Y'_n] \approx \Gamma(1 - \gamma) \iff \mu_{1:n} \approx \Gamma(1 - \gamma) a'_n = \Gamma(1 - \gamma) F^{\inv}\prn{\frac{1}{n}},
\]
whereas for $\gamma = 0$
\begin{equation}\label{eq:mu-approx-min-1}
\E[Y'_n] \approx -\gamma^* \iff \mu_{1:n} \approx - a'_n \gamma^* + b'_n = - a'_n \gamma^* + U_m\prn{\f{1}{n}} = - a'_n \gamma^* + F^{\inv}\prn{\frac{1}{n}}.
\end{equation}
Finally, for $\gamma = 0$, by \cite{dehaan-ferreira-evt} (Theorem 1.1.6), we have that, for any $x > 1$
\begin{equation}\label{eq:mu-approx-min-2}
a'_n \approx \frac{U_m\prn{\f{1}{x n}} - U_m\prn{\f{1}{n}}}{- \log{x}} = \frac{F^{\inv}\prn{\frac{1}{x n}} - F^{\inv}\prn{\frac{1}{n}}}{- \log{x}}.
\end{equation}
Combining \eqref{eq:mu-approx-min-1} and \eqref{eq:mu-approx-min-2}, we get
\[
\mu_{1:n} \approx \frac{\gamma^*}{\log{x}} F^{\inv}\prn{\frac{1}{x n}} + \prn{1 - \frac{\gamma^*}{\log{x}}} F^{\inv}\prn{\frac{1}{n}}.
\]
Setting $x = e^{\gamma^*}$ yields
\[
\mu_{1:n} \approx F^{\inv}\prn{\frac{e^{-\gamma^*}}{n}}.
\]
\end{proof}

The asymptotic expression for $\mu_{n:n}$ and $\mu_{1:n}$ leads us to consider alternative representations of the tail quantiles of $F$.

\begin{lemma}\label{lem:mult-quantile-approx}
For every $F \in D_\gamma$, $c > 0$ and large enough $n$, we have
\begin{itemize}
\item For $\gamma \in (0,1)$,
\[
F^{\inv}\prn{1 - \frac{c}{n}} \approx c^{-\gamma} F^{\inv}\prn{1 - \frac{1}{n}}.
\]
\item For $\gamma < 0$,
\[
x^* - F^{\inv}\prn{1 - \frac{c}{n}} \approx c^{-\gamma} \prn{x^* - F^{\inv}\prn{1 - \frac{1}{n}}},
\]
and
\[
F^{\inv}\prn{\frac{c}{n}} \approx c^{-\gamma} F^{\inv}\prn{\frac{1}{n}}.
\]
\end{itemize}
\end{lemma}
\begin{proof}
Let $\gamma \in (0,1)$. Since $F \in D_\gamma$, we know that $U_M\prn{\f{1}{x}} \in RV_{\gamma}$ (\cite{dehaan-ferreira-evt}, Corollary 1.2.10). By the definition of regular variation, we have
\[
\lim_{n \to \infty} \frac{U_M\prn{\f{1}{x n}}}{U_M\prn{\f{1}{n}}} = x^{\gamma},
\]
for all $x > 0$. Recall that $U_M\prn{\f{1}{x}} = F^{\inv}\prn{1 - \frac{1}{x}}$, and thus
\[
\lim_{n \to \infty} {\frac{F^{\inv}\prn{1 - \frac{1}{x n}}}{F^{\inv}\prn{1 - \frac{1}{n}}}} = x^{\gamma},
\]
for all $x > 0$. For $x = \frac{1}{c}$, we obtain, for large enough $n$
\[
F^{\inv}\prn{1 - \frac{c}{n}} \approx c^{-\gamma} F^{\inv}\prn{1 - \frac{1}{n}}.
\]

Next, let $\gamma < 0$. Since $F \in D_\gamma$, we know that $x^* - U_M\prn{\f{1}{x}}, U_m\prn{\f{1}{x}} \in RV_{\gamma}$ (\cite{dehaan-ferreira-evt}, Corollary 1.2.10). By the definition of regular variation, we have
\[
\lim_{n \to \infty} \frac{x^* - U_M\prn{\f{1}{x n}}}{x^* - U_M\prn{\f{1}{n}}} = \lim_{n \to \infty} \frac{U_m\prn{\f{1}{x n}}}{U_m\prn{\f{1}{n}}} = x^{\gamma},
\]
for all $x > 0$. Recall that $x^* - U_M\prn{\f{1}{x}} = x^* - F^{\inv}\prn{1 - \frac{1}{x}}$ and $U_m\prn{\f{1}{x}} = F^{\inv}\prn{\frac{1}{x}}$. Therefore,
\[
\lim_{n \to \infty} {\frac{x^* - F^{\inv}\prn{1 - \frac{1}{x n}}}{x^* - F^{\inv}\prn{1 - \frac{1}{n}}}} = \lim_{n \to \infty} {\frac{F^{\inv}\prn{\frac{1}{x n}}}{F^{\inv}\prn{\frac{1}{n}}}} = x^{\gamma},
\]
for all $x > 0$. For $x = \frac{1}{c}$, we obtain, for large enough $n$
\[
x^* - F^{\inv}\prn{1 - \frac{c}{n}} \approx c^{-\gamma} \prn{x^* - F^{\inv}\prn{1 - \frac{1}{n}}},
\]
and
\[
F^{\inv}\prn{\frac{c}{n}} \approx c^{-\gamma} F^{\inv}\prn{\frac{1}{n}}.
\]
\end{proof}

Next, we use Lemmas~\ref{lem:mu-approx} and~\ref{lem:mult-quantile-approx} to characterize exactly how the ratio of the prophet's expected value for $n-1$ and $n$ behaves as $n \to \infty$.

\begin{lemma}\label{lem:mu-ratio-approx}
For $F \in D_\gamma$ and large enough $n$, we have
\begin{itemize}
\item For $\gamma \in (0,1)$,
\[
\frac{\mu_{n-1:n-1}}{\mu_{n:n}} = 1 - \frac{\gamma}{n} + \ltlo{\frac{1}{n}}.
\]
\item For $\gamma < 0$,
\[
\frac{x^* - \mu_{n-1:n-1}}{x^* - \mu_{n:n}} = 1 - \frac{\gamma}{n} + \ltlo{\frac{1}{n}}.
\]
and
\[
\frac{\mu_{1:n-1}}{\mu_{1:n}} = 1 - \frac{\gamma}{n} + \ltlo{\frac{1}{n}}.
\]
\end{itemize}
\end{lemma}
\begin{proof}
Let $\gamma \in (0,1)$. Notice that, by Lemma~\ref{lem:mu-approx}, for large enough $n$, we have
\[
\mu_{n-1:n-1} \approx \Gamma(1 - \gamma) F^{\inv}\prn{1 - \frac{1}{n-1}} = \Gamma(1 - \gamma) F^{\inv}\prn{1 - \frac{c}{n}},
\]
for $c = 1 + \frac{1}{n-1}$. Using Lemma~\ref{lem:mult-quantile-approx}, we obtain
\[
\mu_{n-1:n-1} \approx \Gamma(1 - \gamma) \prn{1+\frac{1}{n-1}}^{-\gamma} F^{\inv}\prn{1 - \frac{1}{n}}.
\]
Also, again by Lemma~\ref{lem:mu-approx}, we have
\[
\mu_{n:n} \approx \Gamma(1 - \gamma) F^{\inv}\prn{1 - \frac{1}{n}},
\]
and thus, for large enough $n$
\[
\frac{\mu_{n-1:n-1}}{\mu_{n:n}} \approx \prn{1+\frac{1}{n-1}}^{-\gamma} = 1 - \frac{\gamma}{n} + \ltlo{\frac{1}{n}}.
\]

Similarly, for $\gamma < 0$, by Lemma~\ref{lem:mu-approx}, for large enough $n$, we have
\[
x^* - \mu_{n-1:n-1} \approx \Gamma(1 - \gamma) F^{\inv}\prn{1 - \frac{1}{n-1}} = \Gamma(1 - \gamma) F^{\inv}\prn{1 - \frac{c}{n}},
\]
and
\[
\mu_{1:n-1} \approx \Gamma(1 - \gamma) F^{\inv}\prn{\frac{1}{n-1}} = \Gamma(1 - \gamma) F^{\inv}\prn{\frac{c}{n}},
\]
for $c = 1 + \frac{1}{n-1}$. Using Lemma~\ref{lem:mult-quantile-approx}, we obtain
\[
x^* - \mu_{n-1:n-1} \approx \Gamma(1 - \gamma) \prn{1 +\frac{1}{n-1}}^{-\gamma} F^{\inv}\prn{1 - \frac{1}{n}},
\]
and
\[
\mu_{1:n-1} \approx \Gamma(1 - \gamma) \prn{1+\frac{1}{n-1}}^{-\gamma} F^{\inv}\prn{\frac{1}{n}}.
\]
Also, again by Lemma~\ref{lem:mu-approx}, we have
\[
x^* - \mu_{n:n} \approx \Gamma(1 - \gamma) F^{\inv}\prn{1 - \frac{1}{n}},
\]
and
\[
\mu_{1:n} \approx \Gamma(1 - \gamma) F^{\inv}\prn{\frac{1}{n}},
\]
and thus, for large enough $n$
\[
\frac{x^* - \mu_{n-1:n-1}}{x^* - \mu_{n:n}} \approx \prn{1+\frac{1}{n-1}}^{-\gamma} = 1 - \frac{\gamma}{n} + \ltlo{\frac{1}{n}},
\]
and
\[
\frac{\mu_{1:n-1}}{\mu_{1:n}} \approx \prn{1+\frac{1}{n-1}}^{-\gamma} = 1 - \frac{\gamma}{n} + \ltlo{\frac{1}{n}}.
\]
\end{proof}

Successive applications of the lemma above yield the following corollary, which may be of independent interest.
\begin{corollary}\label{cor:mu-ratio-approx-gen}
For $F \in D_\gamma$, large enough $n$ and $m < n$, we have
\begin{itemize}
\item For $\gamma \in (0,1)$,
\[
\frac{\mu_{m:m}}{\mu_{n:n}} = \frac{\Gamma(m+1)}{\Gamma(n+1)} \cdot \frac{\Gamma(n+1 - \gamma)}{\Gamma(m+1-\gamma)} + \ltlo{\frac{1}{n}}.
\]
\item For $\gamma < 0$,
\[
\frac{x^* - \mu_{m:m}}{x^* - \mu_{n:n}} = \frac{\Gamma(m+1)}{\Gamma(n+1)} \cdot \frac{\Gamma(n+1 - \gamma)}{\Gamma(m+1-\gamma)} + \ltlo{\frac{1}{n}}.
\]
and
\[
\frac{\mu_{1:m}}{\mu_{1:n}} = \frac{\Gamma(m+1)}{\Gamma(n+1)} \cdot \frac{\Gamma(n+1 - \gamma)}{\Gamma(m+1-\gamma)} + \ltlo{\frac{1}{n}}.
\]
\end{itemize}
In particular, taking the series expansion of $\Gamma(\cdot)$ around infinity, it is easy to see that
\begin{align*}
\frac{\Gamma(m+1)}{\Gamma(n+1)} \cdot \frac{\Gamma(n+1 - \gamma)}{\Gamma(m+1-\gamma)} &= 1 - \frac{k \gamma}{n} + \ltlo{\frac{1}{n}}, \text{ for } m = n - k \\
\frac{\Gamma(m+1)}{\Gamma(n+1)} \cdot \frac{\Gamma(n+1 - \gamma)}{\Gamma(m+1-\gamma)} &= \frac{1}{k^\gamma} - \frac{2 \gamma (\gamma - 1)}{k^\gamma n} + \ltlo{\frac{1}{n}}, \text{ for } m = \f{n}{k}.
\end{align*}
\end{corollary}

We are finally ready to prove Theorem~\ref{thm:unified-acr-case-1}.

\begin{proof}[Proof of Theorem~\ref{thm:unified-acr-case-1}]
Recall that $\lambda_M(n)$ and $\lambda_m(n)$ denote the competitive ratio of the I.I.D. Max-Prophet Inequality and I.I.D. Min-Prophet Inequality settings respectively.

For large enough $n$, we have
\begin{align*}
\lambda_M(n) &= \frac{G_M(n)}{\mu_{n:n}} \approx \frac{G_M(n-1)}{\mu_{n:n}} \prn{1 - e^{-H(G_M(n-1))}\prn{1 - \frac{1}{1-\gamma}}} \\
&= \lambda_M(n-1) \frac{\mu_{n-1:n-1}}{\mu_{n:n}} \prn{1 - e^{-H(\lambda_M(n-1) \mu_{n-1:n-1})}\prn{1 - \frac{1}{1-\gamma}}},
\end{align*}
and
\begin{align*}
\lambda_m(n) &= \frac{G_m(n)}{\mu_{1:n}} \approx \frac{G_m(n-1)}{\mu_{1:n}} \prn{1 - e^{R(G_m(n-1))}\prn{1 - \frac{1}{1-\gamma}}} \\
&= \lambda_m(n-1) \frac{\mu_{1:n-1}}{\mu_{1:n}} \prn{1 - e^{R(\lambda_m(n-1) \mu_{1:n-1})}\prn{1 - \frac{1}{1-\gamma}}},
\end{align*}
by Lemma~\ref{lem:gn-done}. Using Lemma~\ref{lem:mu-ratio-approx}, we obtain
\begin{equation}\label{eq:new-lemma-sub-01}
\lambda_M(n) = \lambda_M(n-1) \prn{1 - \frac{\gamma}{n} + \ltlo{\frac{1}{n}}} \prn{1 - e^{-H(\lambda_M(n-1) \mu_{n-1:n-1})}\prn{1 - \frac{1}{1-\gamma}}},
\end{equation}
and
\begin{equation}\label{eq:new-lemma-sub-02}
\lambda_m(n) = \lambda_m(n-1) \prn{1 - \frac{\gamma}{n} + \ltlo{\frac{1}{n}}} \prn{1 - e^{R(\lambda_m(n-1) \mu_{1:n-1})}\prn{1 - \frac{1}{1-\gamma}}}.
\end{equation}
Next, by \cite{klein-wein} and Theorem~\ref{thm:single-threshold-evt}, we have that
\begin{equation}\label{eq:new-lemma-sub-1}
\lambda_M(n) = \bigO{1} \text{ and } \lambda_M(n) = \bigO{\prn{\log{n}}^{-\gamma}},
\end{equation}
and thus
\begin{equation}\label{eq:new-lemma-sub-2}
\lambda_M(n) - \lambda_M(n-1) \approx 0 \text{ and }
\lambda_m(n) - \lambda_m(n-1) \approx 0.
\end{equation}
To see why the asymptotic equality for $\lambda_m$ is true, notice that there exists a constant $c > 0$ such that
\begin{align*}
\prn{\log(n+1)}^{-\gamma} - \prn{\log(n)}^{-\gamma} &\leq \prn{\log(n+1)}^{\ceil{-\gamma}} - \prn{\log(n)}^{\ceil{-\gamma}} \\
&\leq \prn{\log(n+1) - \log(n)} \sum_{j = 1}^{\ceil{-\gamma}} \prn{\log(n+1)}^{\ceil{-\gamma} - j} \prn{\log(n)}^{j-1} \\
&\leq \log\prn{\frac{n+1}{n}} \cdot c \ceil{-\gamma} \prn{\log{n}}^{\ceil{-\gamma}} \\
&\leq \log\prn{1 + \frac{1}{n}} \cdot c \ceil{-\gamma} \prn{\log{n}}^{\ceil{-\gamma}} \\
&\leq \prn{\frac{1}{n} + \ltlo{\frac{1}{n}}} \cdot c \ceil{-\gamma} \prn{\log{n}}^{\ceil{-\gamma}} \\
&\approx 0,
\end{align*}
where the first inequality follows from the monotonicity of $\log(\cdot)$ and the last inequality follows from the series expansion of $\log(1+z)$ around $z = 0$. Therefore, by \eqref{eq:new-lemma-sub-01}, \eqref{eq:new-lemma-sub-02}, \eqref{eq:new-lemma-sub-1} and \eqref{eq:new-lemma-sub-2}, we have
\begin{align*}
\lambda_M(n-1) \prn{1 - \prn{1 - \frac{\gamma}{n} + \ltlo{\frac{1}{n}}} \prn{1 - e^{-H(\lambda_M(n-1) \mu_{n-1:n-1})}\prn{1 - \frac{1}{1-\gamma}}}} \approx 0,
\end{align*}
and
\begin{align*}
\lambda_m(n-1) \prn{1 - \prn{1 - \frac{\gamma}{n} + \ltlo{\frac{1}{n}}} \prn{1 - e^{R(\lambda_m(n-1) \mu_{1:n-1})}\prn{1 - \frac{1}{1-\gamma}}}} \approx 0.
\end{align*}
Since $\lambda_M, \lambda_m \neq 0$,
\begin{align*}
\prn{1 - \frac{\gamma}{n} + \ltlo{\frac{1}{n}}} \prn{1 - e^{-H(\lambda_M(n-1) \mu_{n-1:n-1})}\prn{1 - \frac{1}{1-\gamma}}}\approx 1,
\end{align*}
and
\begin{align*}
\prn{1 - \frac{\gamma}{n} + \ltlo{\frac{1}{n}}} \prn{1 - e^{R(\lambda_m(n-1) \mu_{1:n-1})}\prn{1 - \frac{1}{1-\gamma}}} \approx 1,
\end{align*}
and by rearranging terms and ignoring lower-order terms, we get
\[
e^{-H(\lambda_M(n-1) \mu_{n-1:n-1})}\prn{1 - \frac{1}{1-\gamma}} \approx \frac{-\gamma}{n},
\]
and
\[
e^{R(\lambda_m(n-1) \mu_{1:n-1})}\prn{1 - \frac{1}{1-\gamma}} \approx \frac{-\gamma}{n}.
\]
Therefore,
\[
e^{-H(\lambda_M(n-1) \mu_{n-1:n-1})} \approx \frac{1-\gamma}{n} \iff H(\lambda_M(n-1) \mu_{n-1:n-1}) \approx -\log\prn{\frac{1-\gamma}{n}},
\]
and
\[
e^{R(\lambda_m(n-1) \mu_{1:n-1})} \approx \frac{1-\gamma}{n} \iff R(\lambda_m(n-1) \mu_{1:n-1}) \approx \log\prn{\frac{1-\gamma}{n}}.
\]
Taking the inverses of $H$ and $R$, we obtain
\[
\lambda_M(n-1) \approx \frac{H^{\inv\prn{-\log\prn{\frac{1-\gamma}{n}}}}}{\mu_{n-1:n-1}},
\]
and
\[
\lambda_m(n-1) \approx \frac{R^{\inv\prn{\log\prn{\frac{1-\gamma}{n}}}}}{\mu_{1:n-1}}.
\]
Since $H^{\inv}\prn{-\log{x}} = F^{\inv}(1 - x)$ and $R^{\inv}\prn{\log{x}} = F^{\inv}(x)$, we obtain
\[
\lambda_M(n-1) \approx \frac{F^{\inv}\prn{1 - \frac{1-\gamma}{n}}}{\mu_{n-1:n-1}},
\]
and
\[
\lambda_m(n-1) \approx \frac{F^{\inv}\prn{\frac{1-\gamma}{n}}}{\mu_{1:n-1}}.
\]
By Lemma~\ref{lem:mu-approx}, we get
\[
\lambda_M(n-1) \approx \frac{F^{\inv}\prn{1 - \frac{1-\gamma}{n}}}{\Gamma(1 - \gamma) \prn{1+\frac{1}{n-1}}^{-\gamma} F^{\inv}\prn{1 - \frac{1}{n}}},
\]
and
\[
\lambda_m(n-1) \approx \frac{F^{\inv}\prn{\frac{1-\gamma}{n}}}{\Gamma(1 - \gamma) \prn{1+\frac{1}{n-1}}^{-\gamma} F^{\inv}\prn{\frac{1}{n}}}.
\]
Finally, using Lemma~\ref{lem:mult-quantile-approx}, we get
\[
\lambda_M(n-1) \approx \frac{\prn{1 - \gamma}^{-\gamma}}{\Gamma(1 - \gamma) \prn{1+\frac{1}{n-1}}^{-\gamma}} \cdot \frac{F^{\inv}\prn{1 - \frac{1}{n}}}{F^{\inv}\prn{1 - \frac{1}{n}}} = \frac{\prn{1 - \gamma}^{-\gamma}}{\Gamma(1 - \gamma) \prn{1+\frac{1}{n-1}}^{-\gamma}},
\]
and
\[
\lambda_m(n-1) \approx \frac{\prn{1 - \gamma}^{-\gamma}}{\Gamma(1 - \gamma) \prn{1+\frac{1}{n-1}}^{-\gamma}} \cdot \frac{F^{\inv}\prn{\frac{1}{n}}}{F^{\inv}\prn{\frac{1}{n}}} = \frac{\prn{1 - \gamma}^{-\gamma}}{\Gamma(1 - \gamma) \prn{1+\frac{1}{n-1}}^{-\gamma}}.
\]
Since $\lambda_M(n-1) \to \lambda_M$ and $\lambda_m(n-1) \to \lambda_m$ as $n$ goes to infinity, we have
\[
\lambda_M = \frac{\prn{1 - \gamma}^{-\gamma}}{\Gamma(1 - \gamma)}, \text{ for } \gamma \in (0,1)
\qquad \text{and} \qquad
\lambda_m = \frac{\prn{1 - \gamma}^{-\gamma}}{\Gamma(1 - \gamma)}, \text{ for } \gamma < 0.
\]
\end{proof}

\subsection{MHR Distributions}\label{sec:mhr}

In this section, we show that for the special case of MHR distributions, $\gamma \leq 0$ for maxima and $\gamma \geq -1$ for minima. Combining this with Theorem~\ref{thm:unified-acr}, for the I.I.D. Max-Prophet Inequality with an MHR distribution, we recover that the asymptotic competitive ratio is $1$, a result of \cite{kesselheim-mhr-ppm} and for the I.I.D. Min-Prophet Inequality, since $\Lambda(-1) = 2$, we recover that the asymptotic competitive ratio is at most $2$, a result of \cite{liv-mehta-cpi}.

\begin{reptheorem}{thm:mhr-acr}
Let $\cD$ be an MHR distribution in the domain of attraction of $D_\gamma$. Then, 
for the I.I.D. Max-Prophet Inequality we have $\gamma \leq 0$ and for the I.I.D. Min-Prophet Inequality we have $\gamma \geq -1$.
\end{reptheorem}
\begin{proof}
First, let $F \in D_\gamma$ for maxima. Then, we have that the hazard rate $h(x)$ is monotonically non-decreasing, which implies that $H(x) = \bigOm{x}$. Therefore, $H^{\inv}(x) = \bigO{x}$ and in particular, $U_M(\f{1}{x}) = H^{\inv}(\log{x}) = \bigO{\log{x}}$.

Assume towards contradiction that $\gamma > 0$. Then, by Lemma~\ref{lem:rv-hazard}, it must be that
\begin{equation}\label{eq:mhr-maxima}
\lim_{t \to \infty} \frac{U_M(\f{1}{tx})}{U_M(\f{1}{t})} = x^\gamma,
\end{equation}
for all $x > 0$. However, we have that
\[
\lim_{t \to \infty} \frac{U_M(\f{1}{tx})}{U_M(\f{1}{t})} \leq \lim_{t \to \infty} \frac{\log{xt}}{\log{t}} = 1,
\]
and thus \eqref{eq:mhr-maxima} does not hold and we arrive at a contradiction.

Next, let $F \in D_\gamma$ for minima. Then, we have that $H(x) = \bigOm{x}$. Notice, however, that $e^{-H(x)} + e^{R(x)} = 1 - F(x) + F(x) = 1$, which implies that
\[
R(x) = \log\prn{1 - e^{-H(x)}},
\]
and thus $R(x) = \log\prn{1 - e^{-\bigOm{x}}}$ for every MHR distribution. Therefore, we have
\[
R^{\inv}\prn{x} = \bigO{-\log{1 - e^x}} \implies U_m(x) = R^{\inv}\prn{\log{x}} = \bigO{-\log{1 - x}}.
\]
Assume towards contradiction that $\gamma < -1 \implies -\gamma > 1$. Then, by Lemma~\ref{lem:rv-hazard}, it must be that
\begin{equation}\label{eq:mhr-minima}
\lim_{t \to 0^+} \frac{U_m(tx)}{U_m(t)} = x^{-\gamma},
\end{equation}
for all $x > 0$. However, we have that
\[
\lim_{t \to 0^+} \frac{U_m(tx)}{U_m(t)} \leq \lim_{t \to 0^+} \frac{-\log{1-xt}}{-\log{1-t}} = x,
\]
and thus \eqref{eq:mhr-minima} does not hold and we arrive at a contradiction.
\end{proof}

\section{Single-Threshold Algorithms for the I.I.D. Min-Prophet Inequality}\label{sec:single-threshold}

In this section, we investigate single-threshold algorithms for the I.I.D. Min-Prophet Inequality and prove Theorem~\ref{thm:single-threshold-evt}. Using Extreme Value Theory, we generalize the results of \cite{liv-mehta-cpi} by designing, for any distribution $F \in D_\gamma$, a single-threshold that achieves a poly-logarithmic competitive ratio, where the exponent is exactly $-\gamma$. We also show that this dependence is optimal up to constants.

\begin{reptheorem}{thm:single-threshold-evt}
For the I.I.D. Min-Prophet Inequality and any distribution $F \in D_\gamma$, where $\gamma < 0$, there exists a single-threshold algorithm that achieves a competitive ratio of $\bigO{\prn{\log{n}}^{-\gamma}}$. Furthermore, this factor is tight; no single-threshold algorithm can achieve a $\ltlo{\prn{\log{n}}^{-\gamma}}$ competitive ratio for all instances.
\end{reptheorem}

Our algorithm sets a fixed threshold $T$ and selects the first realization that is below $T$. If our algorithm ever reaches $X_n$ and has not selected any value, it is forced to pick the realization of $X_n$ regardless of its cost. Recall that $F \in D_\gamma$ for some $\gamma < 0$. Our choice of $T$ is
\[
T = \bigTh{\prn{\frac{\log{n}}{n}}^{-\gamma}}.
\]

For clarity of presentation, we split the analysis of the upper and lower bounds on the competitive ratio into different sections, namely Section~\ref{sec:single-threshold-upper} and Section~\ref{sec:single-threshold-lower}. Theorem~\ref{thm:single-threshold-evt} follows by Theorems~\ref{thm:single-threshold-evt-upper} and~\ref{thm:single-threshold-evt-lower}.

\begin{remark}
Theorem~\ref{thm:single-threshold-evt} holds for $\gamma < 0$. This is because for $\gamma = 0$, as we will see in the proof of Theorem~\ref{thm:single-threshold-evt-upper}, we have that $\mu_{1:1} \leq c \mu_{1:n}$ for some constant $c > 1$ and all $n$. Thus, an algorithm that sets a single threshold $T \in (\mu_{1:n}, \mu_{1:1}]$ achieves a constant competitive ratio.
\end{remark}

Finally, we consider the multiple-selection prophet inequality setting, where the algorithm has to select at least $k$ values, and we show that there exists a single-threshold algorithm that is constant-competitive with respect to $OPT$ for any $k = \bigOm{\log{n}}$. The following theorem is shown in Section~\ref{sec:single-threshold-multi}.

\begin{reptheorem}{thm:single-threshold-evt-multi}
For the $k$-multiple-selection I.I.D. Min-Prophet Inequality with $k = \bigOm{\log{n}}$, and any distribution $F \in D_\gamma$, where $\gamma \leq 0$, there exists a single-threshold algorithm that achieves a constant competitive ratio.
\end{reptheorem}

\subsection{Upper Bound}\label{sec:single-threshold-upper}

\begin{theorem}\label{thm:single-threshold-evt-upper}
For the I.I.D. Min-Prophet Inequality and any distribution $F \in D_\gamma$, there exists a single-threshold $T = T(n, \gamma, F)$ such that the algorithm that selects the first value $X_i \leq T$ for $i < n$ and $X_n$ otherwise, achieves a $\bigO{\prn{\log{n}}^{-\gamma}}$-competitive ratio, for large enough $n$.
\end{theorem}
\begin{proof}
We start by analyzing the algorithm's performance for an arbitrary choice of $T$. We have
\begin{align*}
\E[\alg] &= \prn{1 - \prn{1 - F(T)}^{n-1}} \E\brk{X \midd X \leq T} + \prn{1 - F(T)}^{n-1} \E[X] \\
&= \frac{1 - \prn{1 - F(T)}^{n-1}}{F(T)} \int^T_0 {\prn{F(T) - F(x)} \dif x} +
\prn{1 - F(T)}^{n-1} \mu_{1:1} \\
&= \frac{1 - \prn{1 - F(T)}^{n-1}}{F(T)} \prn{T F(T) - \int^T_0 {F(x) \dif x}} +
\prn{1 - F(T)}^{n-1} \mu_{1:1}.
\end{align*}
Let $t = F(x) \implies \dif x = \prn{F^{\inv}(t)}' \dif t$. Thus,
\begin{align*}
\E[\alg] &= \frac{1 - \prn{1 - F(T)}^{n-1}}{F(T)} \prn{T F(T) - \int^{F(T)}_0 {t \prn{F^{\inv}(t)}' \dif t}} + \prn{1 - F(T)}^{n-1} \mu_{1:1} \\
&= \frac{1 - \prn{1 - F(T)}^{n-1}}{F(T)} \prn{T F(T) - T F(T) + \int^{F(T)}_0 {F^{\inv}(t) \dif t}} + \prn{1 - F(T)}^{n-1} \mu_{1:1} \\
&= \frac{1 - \prn{1 - F(T)}^{n-1}}{F(T)} \cdot \int^{F(T)}_0 {F^{\inv}(t) \dif t} + \prn{1 - F(T)}^{n-1} \mu_{1:1}.
\end{align*}
However, notice that, since $F \in D_\gamma$, by \cite{dehaan-ferreira-evt} (Corollary 1.2.10), we have that $F^{\inv} \in RV_{-\gamma}$. Thus, by Lemma~\ref{lem:karamatas}, we have
\[
\int^{F(T)}_0 {F^{\inv}(t)} \approx \frac{T F(T)}{1 - \gamma},
\]
as $T \to 0^+$, which for our choice of $T$ corresponds to $n \to +\infty$. 

Also, using Lemma~\ref{lem:mu-ratio-approx} and ignoring lower order terms, we have
\[
\mu_{1:1} \approx \prod_{j = 2}^n {\prn{1 - \frac{\gamma}{j}}} \mu_{1:n} = \frac{\Gamma(n+1-\gamma)}{(1-\gamma)\Gamma(1-\gamma) \Gamma(n+1)} \mu_{1:n} = \prn{\frac{n^{-\gamma}}{(1-\gamma)\Gamma(1-\gamma)} + \ltlo{n^{-\gamma}}} \mu_{1:n}.
\]
Therefore, for large enough $n$,
\begin{align*}
\E[\alg] &\approx \frac{1 - \prn{1 - F(T)}^{n-1}}{F(T)} \cdot \frac{T F(T)}{1-\gamma} + \prn{1 - F(T)}^{n-1} \prn{\frac{n^{-\gamma}}{(1-\gamma)\Gamma(1-\gamma)} + \ltlo{n^{-\gamma}}} \mu_{1:n} \\
&= \prn{1 - \prn{1 - F(T)}^{n-1}} \frac{T}{1-\gamma} + \prn{1 - F(T)}^{n-1} \prn{\frac{n^{-\gamma}}{(1-\gamma)\Gamma(1-\gamma)} + \ltlo{n^{-\gamma}}} \mu_{1:n}. \\
\end{align*}
Thus, the competitive ratio is
\begin{equation}\label{eq:single-thr-ratio}
\lambda_m(n) \approx \prn{1 - \prn{1 - F(T)}^{n-1}} \frac{T}{(1-\gamma) \mu_{1:n}} + \prn{1 - F(T)}^{n-1} \prn{\frac{n^{-\gamma}}{(1-\gamma)\Gamma(1-\gamma)} + \ltlo{n^{-\gamma}}}.
\end{equation}
Next, let $T = F^{\inv}\prn{\frac{g(n,\gamma)}{n}}$, for some appropriate function $g$ to be defined later. Then, we obtain
\[
\lambda_m(n) \approx \prn{1 - \prn{1 - \frac{g(n,\gamma)}{n}}^{n-1}} \frac{F^{\inv}\prn{\frac{g(n,\gamma)}{n}}}{(1-\gamma) \mu_{1:n}} + \prn{1 - \frac{g(n,\gamma)}{n}}^{n-1} \prn{\frac{n^{-\gamma}}{(1-\gamma)\Gamma(1-\gamma)} + \ltlo{n^{-\gamma}}}.
\]
Using Lemmas~\ref{lem:mu-approx} and~\ref{lem:mult-quantile-approx}, along with the fact that $\Gamma(2-\gamma) = (1-\gamma) \Gamma(1-\gamma)$, for large enough $n$ we get
\begin{align}
\lambda_m(n) &\approx \prn{1 - \prn{1 - \frac{g(n,\gamma)}{n}}^{n-1}} \frac{F^{\inv}\prn{\frac{g(n,\gamma)}{n}}}{(1-\gamma) \Gamma(1 - \gamma) F^{\inv}\prn{\frac{1}{n}}} + \prn{1 - \frac{g(n,\gamma)}{n}}^{n-1} \prn{\frac{n^{-\gamma}}{(1-\gamma)\Gamma(1-\gamma)} + \ltlo{n^{-\gamma}}} \nonumber \\
&\approx \prn{1 - e^{-g(n,\gamma)}} \frac{\prn{g(n,\gamma)}^{-\gamma} F^{\inv}\prn{\frac{1}{n}}}{(1-\gamma) \Gamma(1 - \gamma) F^{\inv}\prn{\frac{1}{n}}} + e^{-g(n,\gamma)} \prn{\frac{n^{-\gamma}}{(1-\gamma)\Gamma(1-\gamma)} + \ltlo{n^{-\gamma}}} \nonumber \\
\label{eq:single-thr-choice-g} &\approx \frac{1}{\Gamma(2-\gamma)} \prn{\prn{1 - e^{-g(n,\gamma)}} \prn{g(n,\gamma)}^{-\gamma} + e^{-g(n,\gamma)} \prn{n^{-\gamma} + \ltlo{n^{-\gamma}}}}.
\end{align}
Finally, let $g(n,\gamma) = -\gamma \log\prn{\frac{n}{\log{n}}}$. This implies that $e^{-g(n,\gamma)} = \prn{\frac{\log{n}}{n}}^{-\gamma}$ and that $\prn{g(n,\gamma)}^{-\gamma} = \prn{-\gamma}^{-\gamma} \prn{\log\prn{\frac{n}{\log{n}}}}^{-\gamma}$. Therefore, \eqref{eq:single-thr-choice-g} becomes
\begin{align*}
\lambda_m(n) &\approx \frac{1}{\Gamma(2-\gamma)} \prn{\prn{1 - \prn{\frac{\log{n}}{n}}^{-\gamma}} \prn{-\gamma}^{-\gamma} \prn{\log\prn{\frac{n}{\log{n}}}}^{-\gamma} + \prn{\frac{\log{n}}{n}}^{-\gamma} \prn{n^{-\gamma} + \ltlo{n^{-\gamma}}}} \\
&\approx \frac{1}{\Gamma(2-\gamma)} \prn{\prn{1 - \prn{\frac{\log{n}}{n}}^{-\gamma}} \prn{-\gamma}^{-\gamma} \prn{\log\prn{\frac{n}{\log{n}}}}^{-\gamma} + \prn{\log{n}}^{-\gamma} + \ltlo{\prn{\log{n}}^{-\gamma}} } \\
&\approx \frac{\prn{\prn{-\gamma}^{-\gamma} + 1}}{\Gamma(2-\gamma)} \prn{\log{n}}^{-\gamma} + \ltlo{\prn{\log{n}}^{-\gamma}}.
\end{align*}
Thus, our choice of $T$ achieves a competitive ratio of $\bigO{\prn{\log{n}}^{-\gamma}}$.
\end{proof}

\subsection{Lower Bound}\label{sec:single-threshold-lower}

\begin{theorem}\label{thm:single-threshold-evt-lower}
For the I.I.D. Min-Prophet Inequality and any $\gamma \leq 0$, consider the distribution $\cD$ for which $F(x) = 1 - e^{-x^{-\f{1}{\gamma}}}$. For $\cD$ and large enough $n$, no single-threshold algorithm is $\ltlo{\prn{\log{n}}^{-\gamma}}$-competitive.
\end{theorem}
\begin{proof}
Recall by \eqref{eq:single-thr-ratio} that for large $n$
\begin{align*}
\lambda_m(n) &\approx \prn{1 - \prn{1 - F(T)}^{n-1}} \frac{T}{(1-\gamma) \mu_{1:n}} + \prn{1 - F(T)}^{n-1} \prn{\frac{n^{-\gamma}}{(1-\gamma)\Gamma(1-\gamma)} + \ltlo{n^{-\gamma}}} \\
&\approx \frac{1}{\Gamma(2-\gamma)} \prn{\prn{1 - e^{-(n-1)T^{-\f{1}{\gamma}}}} \frac{T \Gamma(1-\gamma)}{\mu_{1:n}} + e^{-(n-1)T^{-\f{1}{\gamma}}} \prn{n^{-\gamma} + \ltlo{n^{-\gamma}}}}.
\end{align*}
Next, notice that
\[
\mu_{1:n} = \int^\infty_0 {e^{-n x^{-\f{1}{\gamma}}} \dif x} = \frac{\Gamma(1-\gamma)}{n^{-\gamma}}.
\]
Thus, ignoring lower order terms, we have that
\begin{align*}
\lambda_m(n) &\approx \frac{n^{-\gamma}}{\Gamma(2-\gamma)} \prn{\prn{1 - e^{-(n-1)T^{-\f{1}{\gamma}}}} T + e^{-(n-1)T^{-\f{1}{\gamma}}}}.
\end{align*}
Assume, towards contradiction, that $\lambda_m(n) = \ltlo{\prn{\log{n}}^{-\gamma}}$. For this to be the case, it must be that
\begin{equation}\label{eq:single-lower-bound-1}
\prn{1 - e^{-(n-1)T^{-\f{1}{\gamma}}}} T = \ltlo{\prn{\frac{\log{n}}{n}}^{-\gamma}},
\end{equation}
and also that
\begin{equation}\label{eq:single-lower-bound-2}
e^{-(n-1)T^{-\f{1}{\gamma}}} = \ltlo{\prn{\frac{\log{n}}{n}}^{-\gamma}}.
\end{equation}
By \eqref{eq:single-lower-bound-2} and the definition of $\ltlo{\cdot}$, we have that for every $\eps > 0$, there must exist a $n_0 \geq 1$ such that for all $n \geq n_0$, we have
\begin{align}\label{eq:single-lower-bound-3}
e^{-(n-1)T^{-\f{1}{\gamma}}} \leq \eps \prn{\frac{\log{n}}{n}}^{-\gamma} \iff T \geq \prn{-\frac{\log\prn{\eps \prn{\frac{\log{n}}{n}}^{-\gamma}}}{n-1}}^{-\gamma}.
\end{align}

However, by \eqref{eq:single-lower-bound-1}, we have that for every $\eps' > 0$, there must exist a $n_1 \geq 1$ such that for all $n \geq n_1$, we have
\begin{equation}\label{eq:single-lower-bound-4}
\prn{1 - e^{-(n-1)T^{-\f{1}{\gamma}}}} T \leq \eps' \: \prn{\frac{\log{n}}{n}}^{-\gamma}.
\end{equation}
From \eqref{eq:single-lower-bound-3}, we know that
\[
1 - e^{-(n-1)T^{-\f{1}{\gamma}}} \geq 1 - \eps \prn{\frac{\log{n}}{n}}^{-\gamma}
\]
and thus, by \eqref{eq:single-lower-bound-4} for $\eps' = \eps$, it must be the case that
\begin{align}
\prn{1 - \eps \prn{\frac{\log{n}}{n}}^{-\gamma}} T &\leq \eps \prn{\frac{\log{n}}{n}}^{-\gamma} \iff \nonumber \\
\label{eq:single-lower-bound-5} T &\leq \frac{\eps \prn{\frac{\log{n}}{n}}^{-\gamma}}{1 - \eps \prn{\frac{\log{n}}{n}}^{-\gamma}}.
\end{align}
Thus, using \eqref{eq:single-lower-bound-3} and \eqref{eq:single-lower-bound-5}, we have
\begin{equation}\label{eq:single-lower-bound-6}
\prn{-\frac{\log\prn{\eps \prn{\frac{\log{n}}{n}}^{-\gamma}}}{n-1}}^{-\gamma} \leq \frac{\eps \prn{\frac{\log{n}}{n}}^{-\gamma}}{1 - \eps \prn{\frac{\log{n}}{n}}^{-\gamma} }.
\end{equation}
Let $\eps = \prn{\frac{-\gamma}{2}}^{-\gamma} > 0$ for $\gamma > 0$. Also let $B_\gamma(n) = \prn{\frac{-\gamma}{2} \frac{\log{n}}{n}}^{-\gamma}$. Then, \eqref{eq:single-lower-bound-6} becomes
\begin{align}
-\log{B_\gamma(n)} &\leq (n-1) \prn{\frac{B_\gamma(n)}{1 - B_\gamma(n)}}^{-\f{1}{\gamma}} \iff \nonumber \\
\label{eq:single-lower-bound-7} -\gamma \log\prn{\frac{2n}{(-\gamma)\log{n}}} &\leq (n-1) \frac{\frac{(-\gamma) \log{n}}{2n}}{\prn{1 - \prn{\frac{(-\gamma) \log{n}}{2n}}^{-\gamma} }^{-\f{1}{\gamma}}}.
\end{align}
Taking the Taylor series of $\frac{\frac{(-\gamma) \log{n}}{2n}}{\prn{1 - \prn{\frac{(-\gamma) \log{n}}{2n}}^{-\gamma} }^{-\f{1}{\gamma}}}$ around infinity, we have that
\[
\frac{\frac{(-\gamma) \log{n}}{2n}}{\prn{1 - \prn{\frac{(-\gamma) \log{n}}{2n}}^{-\gamma} }^{-\f{1}{\gamma}}} \approx \frac{-\gamma \log{n}}{2n}, 
\]
and thus \eqref{eq:single-lower-bound-7} becomes
\begin{equation}\label{eq:single-lower-bound-8}
-\gamma \log\prn{\frac{2n}{(-\gamma)\log{n}}} \leq \frac{-\gamma}{2} \cdot \frac{n-1}{n} \log{n} \iff \log\prn{\frac{2n}{(-\gamma)\log{n}}} \leq \frac{1}{2} \cdot \frac{n-1}{n} \log{n}.
\end{equation}
Since
\[
\lim_{n\to\infty} {\frac{\log\prn{\frac{2n}{(-\gamma)\log{n}}}}{\log{n}}} = 1,
\]
we have that, for large enough $n$, \eqref{eq:single-lower-bound-8} does not hold. Therefore, there exists $\eps > 0$ such that for all $n_0 \in \N$ and $n \geq n_0$, \eqref{eq:single-lower-bound-1} and \eqref{eq:single-lower-bound-2} cannot simultaneously hold, and we arrive at a contradiction.
\end{proof}

\subsection{Multiple Selection}\label{sec:single-threshold-multi}

In this section, we show Theorem~\ref{thm:single-threshold-evt-multi}. We assume without loss of generality that $k = \log{n}$; the existence of a constant-competitive single-threshold algorithm for larger values of $k$ follows by monotonicity, since the problem only becomes easier\footnote{Consider that, if $k = n$, then every algorithm is $1$-competitive.}. To show Theorem~\ref{thm:single-threshold-evt-multi}, we first need an asymptotic expression of $OPT$. Notice that, by linearity of expectation, $OPT$ is always the sum of the $k$-th smallest order statistics, and thus
\begin{equation}\label{eq:multi-opt}
OPT = \sum_{i = 1}^k {\mu_{i:n}}.    
\end{equation}
To simplify the expression above, we need the following lemma.

\begin{lemma}\label{lem:kth-order-statistic}
For $F \in D_\gamma$, large enough $n$ and $k < n$, we have
\[
\frac{\mu_{k:n}}{\mu_{1:n}} \approx \frac{k^{-\gamma}}{\Gamma(1-\gamma)}.
\]
\end{lemma}
\begin{proof}
Let $U$ be a random variable following the uniform distribution on $[0,1]$. It is well-known that $F^\inv\prn{U}$ follows $F$\footnote{ One can easily see this by thinking of $U$ as the random quantile of $F$ -- see also (\cite{first-course-order-statistics}, (1.1.3)).}. Thus, $X_{k:n}$ is distributed as $F^\inv\prn{U_{k:n}}$, where $U_{k:n}$ is the $k$-th order statistic of $n$ I.I.D. draws from the uniform $[0,1]$ distribution. By the series approximation of $\mu_{k:n}$ (\cite{first-course-order-statistics}, (5.5.3)), we have
\[
\mu_{k:n} \longrightarrow F^\inv\prn{\frac{k}{n+1}},
\]
as $n \to \infty$. Thus, by Lemma~\ref{lem:mult-quantile-approx},
\[
\mu_{k:n} \approx \prn{k\prn{1 - \frac{1}{n+1}}}^{-\gamma} F^\inv\prn{\frac{1}{n}} \approx \frac{\prn{k\prn{1 - \frac{1}{n+1}}}^{-\gamma}}{\Gamma(1-\gamma)} \mu_{1:n},
\]
and thus
\[
\frac{\mu_{k:n}}{\mu_{1:n}} \approx \frac{k^{-\gamma}}{\Gamma(1-\gamma)}.
\]
\end{proof}

Using Lemmas~\ref{lem:mu-approx} and~\ref{lem:kth-order-statistic}, \eqref{eq:multi-opt} becomes
\begin{equation}\label{eq:multi-opt-2}
OPT = \sum_{i = 1}^k {\frac{k^{-\gamma}}{\Gamma(1-\gamma)} \mu_{1:n}} = k^{1-\gamma} F^\inv\prn{\frac{1}{n}}.
\end{equation}

Next, we compute the expected value of a single threshold algorithm with threshold $T$. Let $S = \set{i \midd X_i \leq T}$. We have
\[
\E[ALG] \leq \sum_{i = 0}^k \Pr[|S| = i] \prn{\sum_{j = 1}^n \mu_{j:n} + (k-i) \mu_{1:1}} + \Pr[|S| > k] k T.
\]
The inequality above follows from the following two facts: $(i)$ if $|S| = i$ for some $i < k$, the algorithm needs to select the last $k-i$ random variables to satisfy the constraint, and $(ii)$ if $|S| > k$, we can upper bound every random variable selected by the algorithm by $T$.

Next, let $T = F^\inv\prn{\frac{c \log{n}}{n}}$ for some appropriate constant $c > 1$ to be chosen later, and notice that, since the $X_i$'s are drawn I.I.D. from $\cD$, we have that $\1\brk{X_i \leq T}$ is a Bernoulli random variable with probability $\frac{c\log{n}}{n}$. Thus
\[
\Pr[|S| = i] = \binom{n}{i} \prn{\frac{c\log{n}}{n}}^i \prn{1 - \frac{c\log{n}}{n}}^{n-i}.
\]
Thus, as $n \to \infty$, $|S|$ converges in distribution to a Poisson random variable $Y$ with rate $c \log{n}$, and, asymptotically,
\begin{align*}
\E[ALG] &\leq e^{-c \log{n}} \sum_{i = 0}^k \frac{\prn{c\log{n}}^i}{i!} \prn{\sum_{j = 1}^i \mu_{j:n} + (k-i) \mu_{1:1}} + e^{-c\log{n}} \sum_{i = k+1}^\infty \frac{\prn{c\log{n}}^i}{i!} k T \\
&= e^{-c \log{n}} \sum_{i = 0}^k \frac{\prn{c\log{n}}^i}{i!} \prn{\sum_{j = 1}^i \mu_{j:n} + \sum_{j = i+1}^k \mu_{1:1}} + e^{-c\log{n}} \sum_{i = k+1}^\infty \frac{\prn{c\log{n}}^i}{i!} k T \\
&\approx e^{-c \log{n}} \sum_{i = 0}^k \frac{\prn{c\log{n}}^i}{i!} \prn{\sum_{j = 1}^i \mu_{j:n} + \sum_{j = i+1}^k \frac{\Gamma(2)\Gamma(n+1-\gamma)}{\Gamma(n+1)\Gamma(2-\gamma)} \mu_{1:n}} \\
&\qquad + e^{-c\log{n}} \sum_{i = k+1}^\infty \frac{\prn{c\log{n}}^i}{i!} k F^\inv\prn{\frac{c\log{n}}{n}} \\
&\approx e^{-c \log{n}} \sum_{i = 0}^k \frac{\prn{c\log{n}}^i}{i!} \prn{\sum_{j = 1}^i \mu_{j:n} + \sum_{j = i+1}^k \frac{\Gamma(n+1-\gamma)}{\Gamma(n+1)\Gamma(2-\gamma)} \cdot \frac{\Gamma(1-\gamma)}{j^{-\gamma}} \mu_{j:n}} \\
&\qquad + e^{-c\log{n}} \sum_{i = k+1}^\infty \frac{\prn{c\log{n}}^i}{i!} k \prn{c\log{n}}^{-\gamma} F^\inv\prn{\frac{1}{n}},
\end{align*}
where the second-to-last asymptotic equality follows from Corollary~\ref{cor:mu-ratio-approx-gen} and the definition of $T$, while the last asymptotic equality follows from Lemmas~\ref{lem:kth-order-statistic} and~\ref{lem:mult-quantile-approx}. Using Corollary~\ref{cor:mu-ratio-approx-gen} again to simplify the above expression, we get
\begin{align*}
\E[ALG] &\leq e^{-c \log{n}} \sum_{i = 0}^k \frac{\prn{c\log{n}}^i}{i!} \prn{\sum_{j = 1}^i \mu_{j:n} + \sum_{j = i+1}^k \frac{\Gamma(n+1-\gamma)}{(1-\gamma) \Gamma(n+1)} \cdot \frac{1}{j^{-\gamma}} \mu_{j:n}} \\
&\qquad + e^{-c\log{n}} \sum_{i = k+1}^\infty \frac{\prn{c\log{n}}^i}{i!} k \prn{c\log{n}}^{-\gamma} F^\inv\prn{\frac{1}{n}} \\
&= e^{-c \log{n}} \sum_{i = 0}^k \frac{\prn{c\log{n}}^i}{i!} \prn{\sum_{j = 1}^i \mu_{j:n} + \frac{1}{1-\gamma} \sum_{j = i+1}^k \prn{\frac{n}{j}}^{-\gamma} \mu_{j:n}} \\
&\qquad + e^{-c\log{n}} \sum_{i = k+1}^\infty \frac{\prn{c\log{n}}^i}{i!} k \prn{c\log{n}}^{-\gamma} F^\inv\prn{\frac{1}{n}} \\
&\leq e^{-c \log{n}} \sum_{i = 0}^k \frac{\prn{c\log{n}}^i}{i!} k \prn{\frac{n}{k}}^{-\gamma} \mu_{k:n} \\
&\qquad + e^{-c\log{n}} \sum_{i = k+1}^\infty \frac{\prn{c\log{n}}^i}{i!} k \prn{c\log{n}}^{-\gamma} F^\inv\prn{\frac{1}{n}} \\
&= k \prn{\frac{n}{k}}^{-\gamma} F^\inv\prn{\frac{k}{n}} e^{-c \log{n}} \sum_{i = 0}^k \frac{\prn{c\log{n}}^i}{i!} \\
&\qquad + k \prn{c\log{n}}^{-\gamma} F^\inv\prn{\frac{1}{n}} e^{-c\log{n}} \sum_{i = k+1}^\infty \frac{\prn{c\log{n}}^i}{i!} \\
&\approx k n^{-\gamma} F^\inv\prn{\frac{1}{n}} e^{-c \log{n}} \sum_{i = 0}^k \frac{\prn{c\log{n}}^i}{i!} \\
&\qquad + k \prn{c\log{n}}^{-\gamma} F^\inv\prn{\frac{1}{n}} e^{-c\log{n}} \sum_{i = k+1}^\infty \frac{\prn{c\log{n}}^i}{i!},
\end{align*}
where the inequality follows from the fact that $\mu_{1:n} \leq \mu_{2:n} \leq \dots \leq \mu_{k:n}$, by definition, and the last asymptotic equality follows again from Lemma~\ref{lem:mult-quantile-approx}.

Next, notice that $e^{-c \log{n}} \sum_{i = 0}^k \frac{\prn{c\log{n}}^i}{i!}$ is equal to $\Pr[Y \leq k]$, and $e^{-c \log{n}} \sum_{i = k+1}^\infty \frac{\prn{c\log{n}}^i}{i!}$ is equal to $\Pr[Y > k]$. Since the CDF of the Poisson distribution with rate $r$ is $\Gamma(k+1, r) / \Gamma(k+1)$, the above inequality becomes
\begin{align*}
\E[ALG] &\leq k n^{-\gamma} F^\inv\prn{\frac{1}{n}} \frac{\Gamma(k+1, c \log{n})}{\Gamma(k+1)} + k \prn{c\log{n}}^{-\gamma} F^\inv\prn{\frac{1}{n}} \prn{1 - \frac{\Gamma(k+1, c \log{n})}{\Gamma(k+1)}} \\
&= k F^\inv\prn{\frac{1}{n}} \prn{n^{-\gamma} \frac{\Gamma(k+1, c \log{n})}{\Gamma(k+1)} + \prn{c\log{n}}^{-\gamma} \prn{1 - \frac{\Gamma(k+1, c \log{n})}{\Gamma(k+1)}}}.
\end{align*}
Combining this with \eqref{eq:multi-opt-2}, we obtain an upper bound on the competitive ratio of
\[
\frac{\E[ALG]}{OPT} \leq \prn{\frac{n}{k}}^{-\gamma} \frac{\Gamma(k+1, c \log{n})}{\Gamma(k+1)} + \prn{\frac{c \log{n}}{k}}^{-\gamma} \prn{1 - \frac{\Gamma(k+1, c \log{n})}{\Gamma(k+1)}}.
\]
The idea of the analysis is to choose $T$ in order to balance these two terms. Notice that, since $k = \log{n}$ and $0 < \Gamma(k+1, c \log{n}) /\Gamma(k+1) < 1$, the second term is upper bounded by a constant, namely $c^{-\gamma}$. In other words,
\begin{equation}\label{eq:multi-unit-cr}
\frac{\E[ALG]}{OPT} \leq \prn{\frac{n}{k}}^{-\gamma} \frac{\Gamma(k+1, c \log{n})}{\Gamma(k+1)} + c^{-\gamma}.
\end{equation}
Therefore, it suffices to show that the first term is also bounded by a constant. By \cite{gamma-facts} (eq. 8.11.7), we know that, as $n \to \infty$,
\[
\Gamma(k+1, c \log{n}) \approx \frac{\prn{c\log{n}}^{k+1} e^{-c \log{n}}}{c\log{n} - k-1}.
\]
Furthermore, by Stirling's approximation to the Gamma function,
\[
\Gamma(k+1) = k\Gamma(k) \approx k \sqrt{\frac{2\pi}{k}} \prn{\frac{k}{e}}^k = \sqrt{2\pi k} \prn{\frac{k}{e}}^k.
\]
Thus, we have
\begin{align*}
\prn{\frac{n}{k}}^{-\gamma} \frac{\Gamma(k+1, c \log{n})}{\Gamma(k+1)} &\approx \prn{\frac{n}{k}}^{-\gamma} \frac{e^k \prn{c\log{n}}^{k+1} e^{-c \log{n}}}{\sqrt{2\pi} \: \: k^{k+1/2} \: \: \prn{c\log{n} - k-1}} \\
&= \frac{c}{\sqrt{2\pi}(c-1)} n^{-\gamma+1-c} \cdot \prn{\frac{1}{\log{n}}}^{-\gamma} \frac{\prn{c\log{n}}^{\log{n}} \cdot \log{n}}{\prn{\log{n}}^{\log{n}+1/2} \cdot \log{n}} \\
&= \frac{c}{\sqrt{2\pi}(c-1)} n^{-\gamma+1-c+\log{c}} \cdot \prn{\log{n}}^{\gamma-1/2}.
\end{align*}
Finally, notice that, for $-\gamma+1-c+\log{c} \leq 0$, the expression asymptotically goes to $0$, as the exponent $\gamma-1/2$ of $\log{n}$ is also negative. Therefore, any $c$ satisfying $c - \log{c} \geq 1 - \gamma$ suffices. Such a $c$ always exists for every $\gamma$, since one can take, for example, $c = e^{1 - \gamma}$ and obtain $e^{1-\gamma} \geq 2(1 - \gamma)$, which is true for all $\gamma \leq 0$.

We conclude that \eqref{eq:multi-unit-cr} is always upper bounded by a constant that depends only on $\gamma$.

\section{Conclusion}\label{sec:conclusion}

In this paper, we study the I.I.D. prophet inequality both for rewards maximization and cost minimization in the large market setting where the given distribution is independent of the number $n$ of random variables, and $n \to +\infty$. We provide a unified analysis of both settings via the lens of extreme value theory that generalizes the results of Kennedy and Kertz \cite{kennedy-kertz}, and show that the optimal asymptotic competitive ratio of both can be described in \emph{closed form} by a single function that depends only on the extreme value index. Our work naturally recovers other prior results by \cite{kesselheim-mhr-ppm,liv-mehta-cpi} on MHR distributions. We also generalize the guarantees of \cite{liv-mehta-cpi} for single-threshold algorithms to the extreme value theory setting, and study the multi-unit version of the setting.

Our results provide a definitive answer to the problem of I.I.D. prophet inequalities for the single-item setting, since a closed form generalization is unlikely outside the realm of extreme value theory. The fact that the optimal competitive ratios can be described by a single function in both the maximization and minimization settings is quite surprising, given the qualitative difference of the guarantees. Further examining more general settings, such as the case of selecting $k$ realizations, from the perspective of extreme value theory is an interesting future direction, and unifying their analysis for both maximization and minimization would bring further clarity to the area of prophet inequalities.

\paragraph{Acknowledgements:} The authors would like to thank the anonymous reviewers for their helpful feedback on an earlier draft of this paper.

\bibliographystyle{alpha}
\bibliography{references}

\newcommand{\etalchar}[1]{$^{#1}$}
\begin{thebibliography}{CFH{\etalchar{+}}21b}

\bibitem[ABF{\etalchar{+}}17]{Adamczyk17}
Marek Adamczyk, Allan Borodin, Diodato Ferraioli, Bart~De Keijzer, and Stefano
  Leonardi.
\newblock Sequential posted-price mechanisms with correlated valuations.
\newblock {\em ACM Trans. Econ. Comput.}, 5(4), 12 2017.

\bibitem[ABN08]{first-course-order-statistics}
Barry~C. Arnold, N.~Balakrishnan, and H.~N. Nagaraja.
\newblock {\em A First Course in Order Statistics}.
\newblock Society for Industrial and Applied Mathematics, 2008.

\bibitem[ACK18]{azar}
Yossi Azar, Ashish Chiplunkar, and Haim Kaplan.
\newblock Prophet secretary: Surpassing the 1-1/e barrier.
\newblock In {\'{E}}va Tardos, Edith Elkind, and Rakesh Vohra, editors, {\em
  Proceedings of the 2018 {ACM} Conference on Economics and Computation,
  Ithaca, NY, USA, June 18-22, 2018}, pages 303--318. {ACM}, 2018.

\bibitem[AEE{\etalchar{+}}17]{abolh}
Melika Abolhassani, Soheil Ehsani, Hossein Esfandiari, MohammadTaghi
  Hajiaghayi, Robert~D. Kleinberg, and Brendan Lucier.
\newblock Beating 1-1/e for ordered prophets.
\newblock In Hamed Hatami, Pierre McKenzie, and Valerie King, editors, {\em
  Proceedings of the 49th Annual {ACM} {SIGACT} Symposium on Theory of
  Computing, {STOC} 2017, Montreal, QC, Canada, June 19-23, 2017}, pages
  61--71. {ACM}, 2017.

\bibitem[AKS21]{AssadiKS}
Sepehr Assadi, Thomas Kesselheim, and Sahil Singla.
\newblock {\em Improved Truthful Mechanisms for Subadditive Combinatorial
  Auctions: Breaking the Logarithmic Barrier}, pages 653--661.
\newblock {SIAM}, 2021.

\bibitem[AKW19]{sample-pi-1}
Pablo~Daniel Azar, Robert Kleinberg, and S.~Matthew Weinberg.
\newblock Prior independent mechanisms via prophet inequalities with limited
  information.
\newblock {\em Games Econ. Behav.}, 118:511--532, 2019.

\bibitem[Ala14]{Alaei14}
Saeed Alaei.
\newblock Bayesian combinatorial auctions: Expanding single buyer mechanisms to
  many buyers.
\newblock {\em SIAM Journal on Computing}, 43(2):930--972, 2014.

\bibitem[AM22]{willma-random}
Nick Arnosti and Will Ma.
\newblock Tight guarantees for static threshold policies in the prophet
  secretary problem.
\newblock In David~M. Pennock, Ilya Segal, and Sven Seuken, editors, {\em {EC}
  '22: The 23rd {ACM} Conference on Economics and Computation, Boulder, CO,
  USA, July 11 - 15, 2022}, page 242. {ACM}, 2022.

\bibitem[AS20]{assad-singla}
Sepehr Assadi and Sahil Singla.
\newblock Improved truthful mechanisms for combinatorial auctions with
  submodular bidders.
\newblock {\em SIGecom Exch.}, 18(1):19--27, 2020.

\bibitem[AW20]{aw-random}
Marek Adamczyk and Michal Wlodarczyk.
\newblock Multi-dimensional mechanism design via random order contention
  resolution schemes.
\newblock {\em SIGecom Exch.}, 17(2):46–53, 1 2020.

\bibitem[BBDS17]{babaioff-pricing}
Moshe Babaioff, Liad Blumrosen, Shaddin Dughmi, and Yaron Singer.
\newblock Posting prices with unknown distributions.
\newblock {\em ACM Trans. Econ. Comput.}, 5(2), 3 2017.

\bibitem[BBK21]{kesselheim-mhr-ppm}
Alexander Braun, Matthias Buttkus, and Thomas Kesselheim.
\newblock Asymptotically optimal welfare of posted pricing for multiple items
  with {MHR} distributions.
\newblock In Petra Mutzel, Rasmus Pagh, and Grzegorz Herman, editors, {\em 29th
  Annual European Symposium on Algorithms, {ESA} 2021, September 6-8, 2021,
  Lisbon, Portugal (Virtual Conference)}, volume 204 of {\em LIPIcs}, pages
  22:1--22:16. Schloss Dagstuhl - Leibniz-Zentrum f{\"{u}}r Informatik, 2021.

\bibitem[BC23]{bubna-chiplunkar}
Archit Bubna and Ashish Chiplunkar.
\newblock Prophet inequality: Order selection beats random order.
\newblock In {\em Proceedings of the 24th ACM Conference on Economics and
  Computation}, EC '23, page 302–336, New York, NY, USA, 2023. Association
  for Computing Machinery.

\bibitem[BCKW15]{pricing-lotteries}
Patrick Briest, Shuchi Chawla, Robert Kleinberg, and S.~Matthew Weinberg.
\newblock Pricing lotteries.
\newblock {\em Journal of Economic Theory}, 156:144--174, 2015.
\newblock Computer Science and Economic Theory.

\bibitem[BGGM12]{bhattacharya-pricing}
Sayan Bhattacharya, Gagan Goel, Sreenivas Gollapudi, and Kamesh Munagala.
\newblock Budget-constrained auctions with heterogeneous items.
\newblock {\em Theory Comput.}, 8(1):429--460, 2012.

\bibitem[BGL{\etalchar{+}}21]{beyhaghi-free-order}
Hedyeh Beyhaghi, Negin Golrezaei, Renato~Paes Leme, Martin P{\'{a}}l, and
  Balasubramanian Sivan.
\newblock Improved revenue bounds for posted-price and second-price mechanisms.
\newblock {\em Oper. Res.}, 69(6):1805--1822, 2021.

\bibitem[BGMS21]{brubach}
Brian Brubach, Nathaniel Grammel, Will Ma, and Aravind Srinivasan.
\newblock Improved guarantees for offline stochastic matching via new ordered
  contention resolution schemes.
\newblock In {\em NeurIPS}, 2021.

\bibitem[BGT87]{bingham-rv}
N.~H. Bingham, C.~M. Goldie, and J.~L. Teugels.
\newblock {\em Regular Variation}.
\newblock Encyclopedia of Mathematics and its Applications. Cambridge
  University Press, 1987.

\bibitem[BH08]{BlumHolen08}
Liad Blumrosen and Thomas Holenstein.
\newblock Posted prices vs. negotiations: An asymptotic analysis.
\newblock In {\em Proceedings of the 9th ACM Conference on Electronic
  Commerce}, EC '08, page~49, New York, NY, USA, 2008. Association for
  Computing Machinery.

\bibitem[CC23]{correa-cristi}
Jos\'{e} Correa and Andr\'{e}s Cristi.
\newblock A constant factor prophet inequality for online combinatorial
  auctions.
\newblock In {\em Proceedings of the 55th Annual ACM Symposium on Theory of
  Computing}, STOC 2023, page 686–697, New York, NY, USA, 2023.

\bibitem[CCD{\etalchar{+}}23]{trading-proph}
José Correa, Andrés Cristi, Paul Dütting, Mohammad Hajiaghayi, Jan Olkowski,
  and Kevin Schewior.
\newblock Trading prophets, 2023.

\bibitem[CD15]{cai-pricing}
Yang Cai and Constantinos Daskalakis.
\newblock Extreme value theorems for optimal multidimensional pricing.
\newblock {\em Games Econ. Behav.}, 92:266--305, 2015.

\bibitem[CDFS22]{sample-pi-2}
Jos{\'{e}}~R. Correa, Paul D{\"u}tting, Felix~A. Fischer, and Kevin Schewior.
\newblock Prophet inequalities for independent and identically distributed
  random variables from an unknown distribution.
\newblock {\em Math. Oper. Res.}, 47(2):1287--1309, 2022.

\bibitem[CDL21]{suchi-small-k}
Shuchi Chawla, Nikhil~R. Devanur, and Thodoris Lykouris.
\newblock Static pricing for multi-unit prophet inequalities (extended
  abstract).
\newblock In Michal Feldman, Hu~Fu, and Inbal Talgam{-}Cohen, editors, {\em Web
  and Internet Economics - 17th International Conference, {WINE} 2021, Potsdam,
  Germany, December 14-17, 2021, Proceedings}, volume 13112 of {\em Lecture
  Notes in Computer Science}, pages 545--546. Springer, 2021.

\bibitem[CFH{\etalchar{+}}18]{correa-survey}
Jos{\'{e}}~R. Correa, Patricio Foncea, Ruben Hoeksma, Tim Oosterwijk, and Tjark
  Vredeveld.
\newblock Recent developments in prophet inequalities.
\newblock {\em SIGecom Exch.}, 17(1):61--70, 2018.

\bibitem[CFH{\etalchar{+}}21a]{correa}
Jos\'{e} Correa, Patricio Foncea, Ruben Hoeksma, Tim Oosterwijk, and Tjark
  Vredeveld.
\newblock Posted price mechanisms and optimal threshold strategies for random
  arrivals.
\newblock {\em Mathematics of Operations Research}, 46(4):1452--1478, 2021.

\bibitem[CFH{\etalchar{+}}21b]{correa-iid}
Jos{\'{e}}~R. Correa, Patricio Foncea, Ruben Hoeksma, Tim Oosterwijk, and Tjark
  Vredeveld.
\newblock Posted price mechanisms and optimal threshold strategies for random
  arrivals.
\newblock {\em Math. Oper. Res.}, 46(4):1452--1478, 2021.

\bibitem[CFPV19]{CorreaPricing}
José Correa, Patricio Foncea, Dana Pizarro, and Victor Verdugo.
\newblock From pricing to prophets, and back!
\newblock {\em Operations Research Letters}, 47(1):25--29, 2019.

\bibitem[CGKM20]{suchi-non-adaptive-graphic-matroid}
Shuchi Chawla, Kira Goldner, Anna~R. Karlin, and J.~Benjamin Miller.
\newblock Non-adaptive matroid prophet inequalities.
\newblock {\em CoRR}, abs/2011.09406, 2020.

\bibitem[CHK07]{chawla07}
Shuchi Chawla, Jason~D. Hartline, and Robert Kleinberg.
\newblock Algorithmic pricing via virtual valuations.
\newblock In {\em Proceedings of the 8th ACM Conference on Electronic
  Commerce}, EC '07, page 243–251, New York, NY, USA, 2007. Association for
  Computing Machinery.

\bibitem[CHMS10]{ChawlaHMS}
Shuchi Chawla, Jason~D. Hartline, David~L. Malec, and Balasubramanian Sivan.
\newblock Multi-parameter mechanism design and sequential posted pricing.
\newblock In {\em Proceedings of the Forty-Second ACM Symposium on Theory of
  Computing}, STOC '10, page 311–320, New York, NY, USA, 2010. Association
  for Computing Machinery.

\bibitem[CL21]{chek-liv}
Chandra Chekuri and Vasilis Livanos.
\newblock On submodular prophet inequalities and correlation gap.
\newblock In Ioannis Caragiannis and Kristoffer~Arnsfelt Hansen, editors, {\em
  Algorithmic Game Theory - 14th International Symposium, {SAGT} 2021, Aarhus,
  Denmark, September 21-24, 2021, Proceedings}, volume 12885 of {\em Lecture
  Notes in Computer Science}, page 410. Springer, 2021.

\bibitem[CLN12]{cloud-computing-application}
Sivadon Chaisiri, Bu-Sung Lee, and Dusit Niyato.
\newblock Optimization of resource provisioning cost in cloud computing.
\newblock {\em IEEE Transactions on Services Computing}, 5(2):164--177, 2012.

\bibitem[CPV21]{jose-dana-victor-evt}
Jos\'{e} Correa, Dana Pizarro, and Victor Verdugo.
\newblock Optimal revenue guarantees for pricing in large markets.
\newblock In {\em Algorithmic Game Theory: 14th International Symposium, SAGT
  2021, Aarhus, Denmark, September 21–24, 2021, Proceedings}, page 221–235,
  Berlin, Heidelberg, 2021. Springer-Verlag.

\bibitem[CSZ20]{correa-random}
Jose Correa, Raimundo Saona, and Bruno Ziliotto.
\newblock Prophet secretary through blind strategies.
\newblock {\em Mathematical Programming}, 08 2020.

\bibitem[DDT14]{daskalakis-intractability-2}
Constantinos Daskalakis, Alan Deckelbaum, and Christos Tzamos.
\newblock The complexity of optimal mechanism design.
\newblock In Chandra Chekuri, editor, {\em Proceedings of the Twenty-Fifth
  Annual {ACM-SIAM} Symposium on Discrete Algorithms, {SODA} 2014, Portland,
  Oregon, USA, January 5-7, 2014}, pages 1302--1318. {SIAM}, 2014.

\bibitem[DDT15]{daskalakis-intractability-1}
Constantinos Daskalakis, Alan Deckelbaum, and Christos Tzamos.
\newblock Strong duality for a multiple-good monopolist.
\newblock In Tim Roughgarden, Michal Feldman, and Michael Schwarz, editors,
  {\em Proceedings of the Sixteenth {ACM} Conference on Economics and
  Computation, {EC} '15, Portland, OR, USA, June 15-19, 2015}, pages 449--450.
  {ACM}, 2015.

\bibitem[DFKL20]{Dutting2}
Paul D{\"u}tting, Michal Feldman, Thomas Kesselheim, and Brendan Lucier.
\newblock Prophet inequalities made easy: Stochastic optimization by pricing
  nonstochastic inputs.
\newblock {\em SIAM Journal on Computing}, 49(3):540--582, 2020.

\bibitem[DKL0]{dutting-combinatorial}
Paul D{\"u}tting, Thomas Kesselheim, and Brendan Lucier.
\newblock An $o(\log \log m)$ prophet inequality for subadditive combinatorial
  auctions.
\newblock {\em SIAM Journal on Computing}, 0(0):FOCS20--239--FOCS20--275, 0.

\bibitem[{\relax DLMF}24]{gamma-facts}
{\it NIST Digital Library of Mathematical Functions}.
\newblock \url{https://dlmf.nist.gov/}, Release 1.2.0 of 2024-03-15, 2024.
\newblock F.~W.~J. Olver, A.~B. {Olde Daalhuis}, D.~W. Lozier, B.~I. Schneider,
  R.~F. Boisvert, C.~W. Clark, B.~R. Miller, B.~V. Saunders, H.~S. Cohl, and
  M.~A. McClain, eds.

\bibitem[Dob21]{Dobzinski-Comb-Auc}
Shahar Dobzinski.
\newblock Breaking the logarithmic barrier for truthful combinatorial auctions
  with submodular bidders.
\newblock {\em {SIAM} J. Comput.}, 50(3), 2021.

\bibitem[DRY15]{revenue-sample-mhr}
Peerapong Dhangwatnotai, Tim Roughgarden, and Qiqi Yan.
\newblock Revenue maximization with a single sample.
\newblock {\em Games and Economic Behavior}, 91:318--333, 2015.

\bibitem[DV16]{Dobzinski-Comb-Auct-Impos}
Shahar Dobzinski and Jan Vondr{\'{a}}k.
\newblock Impossibility results for truthful combinatorial auctions with
  submodular valuations.
\newblock {\em J. {ACM}}, 63(1):5:1--5:19, 2016.

\bibitem[DW12]{daskalakis-mhr}
Constantinos Daskalakis and Seth~Matthew Weinberg.
\newblock Symmetries and optimal multi-dimensional mechanism design.
\newblock In {\em Proceedings of the 13th ACM Conference on Electronic
  Commerce}, EC '12, page 370–387, New York, NY, USA, 2012. Association for
  Computing Machinery.

\bibitem[EFGT22]{ezra}
Tomer Ezra, Michal Feldman, Nick Gravin, and Zhihao~Gavin Tang.
\newblock Prophet matching with general arrivals.
\newblock {\em Mathematics of Operations Research}, 47(2):878--898, 2022.

\bibitem[EHKS18]{EhsaniHKS18}
Soheil Ehsani, MohammadTaghi Hajiaghayi, Thomas Kesselheim, and Sahil Singla.
\newblock Prophet secretary for combinatorial auctions and matroids.
\newblock In Artur Czumaj, editor, {\em Proceedings of the Twenty-Ninth Annual
  {ACM-SIAM} Symposium on Discrete Algorithms, {SODA} 2018, New Orleans, LA,
  USA, January 7-10, 2018}, pages 700--714. {SIAM}, 2018.

\bibitem[EHLM17]{esf-prophsec}
Hossein Esfandiari, MohammadTaghi Hajiaghayi, Vahid Liaghat, and Morteza
  Monemizadeh.
\newblock Prophet secretary.
\newblock {\em SIAM Journal on Discrete Mathematics}, 31(3):1685--1701, 2017.

\bibitem[FGL15]{feldman-combinatorial}
Michal Feldman, Nick Gravin, and Brendan Lucier.
\newblock {\em Combinatorial Auctions via Posted Prices}, pages 123--135.
\newblock {SIAM}, 2015.

\bibitem[FT28]{fisher-tippett-evt}
R.~A. Fisher and L.~H.~C. Tippett.
\newblock Limiting forms of the frequency distribution of the largest or
  smallest member of a sample.
\newblock {\em Mathematical Proceedings of the Cambridge Philosophical
  Society}, 24(2):180–190, 1928.

\bibitem[Gau98]{gamma-book}
Walter Gautschi.
\newblock The incomplete gamma functions since tricomi.
\newblock In {\em In Tricomi's Ideas and Contemporary Applied Mathematics, Atti
  dei Convegni Lincei, n. 147, Accademia Nazionale dei Lincei}, pages 203--237,
  1998.

\bibitem[GKL17]{giannakopoulos-markets}
Yiannis Giannakopoulos, Elias Koutsoupias, and Philip Lazos.
\newblock {Online Market Intermediation}.
\newblock In Ioannis Chatzigiannakis, Piotr Indyk, Fabian Kuhn, and Anca
  Muscholl, editors, {\em 44th International Colloquium on Automata, Languages,
  and Programming (ICALP 2017)}, volume~80 of {\em Leibniz International
  Proceedings in Informatics (LIPIcs)}, pages 47:1--47:14, Dagstuhl, Germany,
  2017. Schloss Dagstuhl--Leibniz-Zentrum fuer Informatik.

\bibitem[GMTS23]{raimundo-random}
Giordano Giambartolomei, Frederik Mallmann-Trenn, and Raimundo Saona.
\newblock Prophet inequalities: Separating random order from order selection,
  2023.

\bibitem[Gne43]{gnedenko-evt}
B.~Gnedenko.
\newblock Sur la distribution limite du terme maximum d'une serie aleatoire.
\newblock {\em Annals of Mathematics}, 44(3):423--453, 1943.

\bibitem[GPZ21]{giannakopoulos-mhr}
Yiannis Giannakopoulos, Diogo Po\c{c}as, and Keyu Zhu.
\newblock Optimal pricing for mhr and $\lambda$-regular distributions.
\newblock {\em ACM Trans. Econ. Comput.}, 9(1), 1 2021.

\bibitem[GW23]{gravin-bipartite}
Nick Gravin and Hongao Wang.
\newblock Prophet inequality for bipartite matching: Merits of being simple and
  nonadaptive.
\newblock {\em Mathematics of Operations Research}, 48(1):38--52, 2023.

\bibitem[Har23]{farouk-poisson-random}
Elfarouk Harb.
\newblock Fishing for better constants: The prophet secretary via
  poissonization, 2023.

\bibitem[HDS24]{discrete-evt-frechet}
Adrien~S. Hitz, Richard~A. Davis, and Gennady Samorodnitsky.
\newblock Discrete extremes.
\newblock {\em Journal of Data Science}, 22(4):524--536, 2024.

\bibitem[HF06]{dehaan-ferreira-evt}
Laurens Haan and Ana Ferreira.
\newblock {\em Extreme Value Theory: An Introduction}.
\newblock Springer, 01 2006.

\bibitem[HK82]{hill-kertz}
T.~P. Hill and Robert~P. Kertz.
\newblock Comparisons of stop rule and supremum expectations of i.i.d. random
  variables.
\newblock {\em Ann. Probab.}, 10(2):336--345, 05 1982.

\bibitem[HKS07]{haji}
Mohammad~Taghi Hajiaghayi, Robert Kleinberg, and Tuomas Sandholm.
\newblock Automated online mechanism design and prophet inequalities.
\newblock In {\em Proceedings of the 22Nd National Conference on Artificial
  Intelligence - Volume 1}, AAAI'07, pages 58--65. AAAI Press, 2007.

\bibitem[HN19]{hart-nisan-menu}
Sergiu Hart and Noam Nisan.
\newblock Selling multiple correlated goods: Revenue maximization and menu-size
  complexity.
\newblock {\em J. Econ. Theory}, 183:991--1029, 2019.

\bibitem[HR09]{hartline-simple}
Jason~D. Hartline and Tim Roughgarden.
\newblock Simple versus optimal mechanisms.
\newblock {\em SIGecom Exch.}, 8(1), 2009.

\bibitem[JMZ22a]{ma-tight-k-prophet}
Jiashuo Jiang, Will Ma, and Jiawei Zhang.
\newblock {\em Tight Guarantees for Multi-unit Prophet Inequalities and Online
  Stochastic Knapsack}, pages 1221--1246.
\newblock {SIAM}, 2022.

\bibitem[JMZ22b]{willma-adversarial}
Jiashuo Jiang, Will Ma, and Jiawei Zhang.
\newblock {\em Tight Guarantees for Multi-unit Prophet Inequalities and Online
  Stochastic Knapsack}, pages 1221--1246.
\newblock {SIAM}, 2022.

\bibitem[JMZ22c]{better-tightness}
Jiashuo Jiang, Will Ma, and Jiawei Zhang.
\newblock Tightness without counterexamples: A new approach and new results for
  prophet inequalities, 2022.

\bibitem[Ker86]{kertz}
Robert~P Kertz.
\newblock Stop rule and supremum expectations of i.i.d. random variables: A
  complete comparison by conjugate duality.
\newblock {\em Journal of Multivariate Analysis}, 19(1):88 -- 112, 1986.

\bibitem[KK91]{kennedy-kertz}
Douglas~P. Kennedy and Robert~P. Kertz.
\newblock {The Asymptotic Behavior of the Reward Sequence in the Optimal
  Stopping of I.I.D. Random Variables}.
\newblock {\em The Annals of Probability}, 19(1):329 -- 341, 1991.

\bibitem[KS77]{kren-such}
Ulrich Krengel and Louis Sucheston.
\newblock Semiamarts and finite values.
\newblock {\em Bull. Amer. Math. Soc.}, 83(4):745--747, 07 1977.

\bibitem[KS78]{kren-such2}
Ulrich Krengel and Louis Sucheston.
\newblock On semiamarts, amarts, and processes with finite value.
\newblock {\em Probability on Banach spaces}, 4:197--266, 1978.

\bibitem[KW19]{klein-wein}
Robert Kleinberg and S.~Matthew Weinberg.
\newblock Matroid prophet inequalities and applications to multi-dimensional
  mechanism design.
\newblock {\em Games Econ. Behav.}, 113:97--115, 2019.

\bibitem[LC22]{evt-in-general}
Paolo Leonetti and Amir~Khorrami Chokami.
\newblock {The maximum domain of attraction of multivariate extreme value
  distributions is small}.
\newblock {\em Electronic Communications in Probability}, 27(none):1 -- 8,
  2022.

\bibitem[LLP{\etalchar{+}}21]{renato-iid}
Allen Liu, Renato~Paes Leme, Martin P{\'{a}}l, Jon Schneider, and
  Balasubramanian Sivan.
\newblock Variable decomposition for prophet inequalities and optimal ordering.
\newblock In P{\'{e}}ter Bir{\'{o}}, Shuchi Chawla, and Federico Echenique,
  editors, {\em {EC} '21: The 22nd {ACM} Conference on Economics and
  Computation, Budapest, Hungary, July 18-23, 2021}, page 692. {ACM}, 2021.

\bibitem[LM24]{liv-mehta-cpi}
Vasilis Livanos and Ruta Mehta.
\newblock {\em Minimization is Harder in the Prophet World}, pages 424--461.
\newblock SIAM, 2024.

\bibitem[Luc17]{lucier-survey}
Brendan Lucier.
\newblock An economic view of prophet inequalities.
\newblock {\em SIGecom Exch.}, 16(1):24–47, September 2017.

\bibitem[Man53]{mandelbrot-distribution}
Benoit Mandelbrot.
\newblock {Contribution {\`a} la th{\'e}orie math{\'e}matique des jeux de
  communication}.
\newblock {\em {Annales de l'ISUP}}, 2(1-2):3--124, 1953.

\bibitem[PRSW22]{tristan}
Tristan Pollner, Mohammad Roghani, Amin Saberi, and David Wajc.
\newblock Improved online contention resolution for matchings and applications
  to the gig economy.
\newblock In David~M. Pennock, Ilya Segal, and Sven Seuken, editors, {\em {EC}
  '22: The 23rd {ACM} Conference on Economics and Computation, Boulder, CO,
  USA, July 11 - 15, 2022}, pages 321--322. {ACM}, 2022.

\bibitem[PT22]{peng-tang}
B.~Peng and Z.~Tang.
\newblock Order selection prophet inequality: From threshold optimization to
  arrival time design.
\newblock In {\em 2022 IEEE 63rd Annual Symposium on Foundations of Computer
  Science (FOCS)}, pages 171--178, Los Alamitos, CA, USA, 11 2022. IEEE
  Computer Society.

\bibitem[QRVW19]{convex-pi}
Junjie Qin, Ram Rajagopal, Shai Vardi, and Adam Wierman.
\newblock Convex prophet inequalities.
\newblock {\em SIGMETRICS Perform. Eval. Rev.}, 46(2):85–86, 1 2019.

\bibitem[QVW24]{convex-pi-2}
Junjie Qin, Shai Vardi, and Adam Wierman.
\newblock Minimization fractional prophet inequalities for sequential
  procurement.
\newblock {\em Mathematics of Operations Research}, 49(2):928--947, 2024.

\bibitem[RBVW13]{gas-supply-market}
Ram Rajagopal, Eilyan Bitar, Pravin Varaiya, and Felix Wu.
\newblock Risk-limiting dispatch for integrating renewable power.
\newblock {\em International Journal of Electrical Power \& Energy Systems},
  44(1):615--628, 2013.

\bibitem[Res14]{resnick-rv}
S.I. Resnick.
\newblock {\em Extreme Values, Regular Variation and Point Processes}.
\newblock Springer, 2014.

\bibitem[RS17]{rs}
Aviad Rubinstein and Sahil Singla.
\newblock Combinatorial prophet inequalities.
\newblock In {\em Proceedings of the Twenty-Eighth Annual ACM-SIAM Symposium on
  Discrete Algorithms}, pages 1671--1687. SIAM, 2017.
\newblock Longer ArXiv version is at \url{http://arxiv.org/abs/1611.00665}.

\bibitem[Rub16]{rubin}
Aviad Rubinstein.
\newblock Beyond matroids: Secretary problem and prophet inequality with
  general constraints.
\newblock In {\em Proceedings of the Forty-Eighth Annual ACM Symposium on
  Theory of Computing}, STOC '16, page 324–332, New York, NY, USA, 2016.
  Association for Computing Machinery.

\bibitem[RWW20]{sample-pi-3}
Aviad Rubinstein, Jack~Z. Wang, and S.~Matthew Weinberg.
\newblock Optimal single-choice prophet inequalities from samples.
\newblock In Thomas Vidick, editor, {\em 11th Innovations in Theoretical
  Computer Science Conference, {ITCS} 2020, January 12-14, 2020, Seattle,
  Washington, {USA}}, volume 151 of {\em LIPIcs}, pages 60:1--60:10. Schloss
  Dagstuhl - Leibniz-Zentrum f{\"{u}}r Informatik, 2020.

\bibitem[SC84]{sam-cahn}
Ester Samuel-Cahn.
\newblock Comparison of threshold stop rules and maximum for independent
  nonnegative random variables.
\newblock {\em The Annals of Probability}, 12(4):1213--1216, 1984.

\bibitem[Shi12]{shimura-discrete-evt}
Takaaki Shimura.
\newblock Discretization of distributions in the maximum domain of attraction.
\newblock {\em Extremes}, 15(3):299--317, Sep 2012.

\bibitem[SM02]{saint-mont}
Uwe Saint-Mont.
\newblock A simple derivation of a complicated prophet region.
\newblock {\em J. Multivar. Anal.}, 80(1):67–72, 1 2002.

\bibitem[SYYZ05]{electricity-markets-1}
Suresh~P. Sethi, Houmin Yan, J.~Houzhong Yan, and Hanqin Zhang.
\newblock An analysis of staged purchases in deregulated time-sequential
  electricity markets.
\newblock {\em Journal of Industrial and Management Optimization},
  1(4):443--463, 2005.

\bibitem[VWB11]{electricity-markets-2}
Pravin~P. Varaiya, Felix~F. Wu, and Janusz~W. Bialek.
\newblock Smart operation of smart grid: Risk-limiting dispatch.
\newblock {\em Proceedings of the IEEE}, 99(1):40--57, 2011.

\bibitem[Yan11]{qiqi}
Qiqi Yan.
\newblock Mechanism design via correlation gap.
\newblock In Dana Randall, editor, {\em Proceedings of the Twenty-Second Annual
  {ACM-SIAM} Symposium on Discrete Algorithms, {SODA} 2011, San Francisco,
  California, USA, January 23-25, 2011}, pages 710--719. {SIAM}, 2011.

\end{thebibliography}

\appendix

\section{Appendix}\label{app:appendix}

\subsection{Omitted Proofs}\label{app:proofs}

\subsubsection{Proof of Lemma~\ref{lem:opt-dp-g}}

\begin{replemma}{lem:opt-dp-g}
For any $n > 1$, we have
\[
G_M(n) = G_M(n-1) + \int^{x^*}_{G_M(n-1)} {\prn{1 - F(u)} \dif u},
\]
and
\[
G_m(n) = \int^{G_m(n-1)}_0 {\prn{1 - F(u)} \dif u}.
\]
\end{replemma}
\begin{proof}
The optimal threshold policy sets a threshold $\tau_i$ when observing $X_i$ equal to $G_M(n-i)$ or $G_m(n-i)$ for the max and min settings, respectively. Therefore, we have
\begin{align*}
G_M(n) &= (1 - F(G_M(n-1))) \E\brk{X | X \geq G_M(n-1)} + F(G_M(n-1)) G_M(n-1) \\
&= (1 - F(G_M(n-1))) \frac{\int_{G_M(n-1)}^{x^*} {u f(u) \dif u}}{(1 - F(G_M(n-1)))} + F(G_M(n-1)) G_M(n-1) \\
&= \E[X] - \int_0^{G_M(n-1)} {u f(u) \dif u} + F(G_M(n-1)) G_M(n-1) \\
&= \int_0^{x^*} {\prn{1 - F(u)} \dif u} - \int_0^{G_M(n-1)} {u \prn{F(u)}' \dif u} + F(G_M(n-1)) G_M(n-1) \\
&= \int_0^{x^*} {\prn{1 - F(u)} \dif u} - \brk{u F(u)}_0^{G_M(n-1)} + \int_0^{G_M(n-1)} {F(u) \dif u} + F(G_M(n-1)) G_M(n-1) \\
&= G_M(n-1) + \int_{G_M(n-1)}^{x^*} {\prn{1 - F(u)} \dif u} - F(G_M(n-1)) G_M(n-1) + F(G_M(n-1)) G_M(n-1) \\
&= G_M(n-1) + \int_{G_M(n-1)}^{x^*} {\prn{1 - F(u)} \dif u}.
\end{align*}
Similarly, for $G_m(n)$, we obtain
\begin{align*}
G_m(n) &= F(G_m(n-1)) \E\brk{X | X \leq G_m(n-1)} + (1 - F(G_m(n-1))) G_m(n-1) \\
&= F(G_m(n-1)) \frac{\int_0^{G_M(n-1)} {u f(u) \dif u}}{F(G_m(n-1))} + (1 - F(G_m(n-1))) G_m(n-1) \\
&= \int_0^{G_M(n-1)} {u \prn{F(u)}' \dif u} + G_m(n-1) - G_m(n-1) F(G_m(n-1)) \\
&= \brk{u F(u)}_0^{G_m(n-1)} - \int_0^{G_M(n-1)} {F(u) \dif u} + G_m(n-1) - G_m(n-1) F(G_m(n-1)) \\
&= G_m(n-1) F(G_m(n-1)) + \int_0^{G_m(n-1)} {1 \dif u} - \int_0^{G_M(n-1)} {F(u) \dif u} - G_m(n-1) F(G_m(n-1)) \\
&= \int_0^{G_m(n-1)} {\prn{1 - F(u)} \dif u}.
\end{align*}
\end{proof}

\subsubsection{Proof of Theorem~\ref{thm:unified-acr-case-2}}

\begin{reptheorem}{thm:unified-acr-case-2}
Let $F \in D_0$. Then, for both the I.I.D. Max-Prophet Inequality and the I.I.D. Min-Prophet Inequality, the asymptotic competitive ratio of the optimal threshold policy as $n \to \infty$ is $1$.
\end{reptheorem}
\begin{proof}
In this case, we will also need the following variant of Karamata's theorem.

\begin{lemma}\label{lem:g-approx-for-zero}
For $F \in D_0$ and large enough $n$, we have
\[
\int^{\exp\prn{-H(G_M(n-1))}}_0 {H^{\inv}\prn{-\log{u}} \dif u} \approx F^{\inv}\prn{1 - \frac{1}{n}} e^{-H(G_M(n-1))},
\]
and
\[
\int^{\exp\prn{R(G_m(n-1))}}_0 {R^{\inv}\prn{\log{u}} \dif u} \approx F^{\inv}\prn{\frac{1}{n}} e^{R(G_m(n-1))},
\]
\end{lemma}
\begin{proof}
The proof follows directly from \cite{dehaan-ferreira-evt} (Corollary 1.2.15), for $t = \exp\set{1+H\prn{G_M(n-1)}}$ and $t = \exp\set{1-R\prn{G_m(n-1)}}$.
\end{proof}

By Lemmas~\ref{lem:g-approx-for-zero} and~\ref{lem:karamatas}, we have that, for large $n$
\[
G_M(n-1) \approx F^{\inv}\prn{1 - \frac{1}{n}} \qquad \text{and} \qquad G_m(n-1) \approx F^{\inv}\prn{\frac{1}{n}}
\]
Therefore, by Lemmas~\ref{lem:gn-done} and \ref{lem:mu-approx}, we have
\[
\lambda_M(n) = \frac{G_M(n)}{\mu_{n:n}} \approx \frac{F^{\inv}\prn{1 - \frac{1}{n}}}{F^{\inv}\prn{1 - \frac{e^{-\gamma^*}}{n}}},
\]
and
\[
\lambda_m(n) = \frac{G_m(n)}{\mu_{1:n}} \approx \frac{F^{\inv}\prn{\frac{1}{n}}}{F^{\inv}\prn{\frac{e^{-\gamma^*}}{n}}},
\]
both of which go to $1$ as $n$ goes to infinity, by Lemma~\ref{lem:rv-hazard}.
\end{proof}

\subsubsection{Proof of Theorem~\ref{thm:unified-acr-case-3}}

\begin{reptheorem}{thm:unified-acr-case-3}
Let $F \in D_\gamma$, for $\gamma < 0$. Then, for the I.I.D. Max-Prophet Inequality, the asymptotic competitive ratio of the optimal threshold policy as $n \to \infty$ is $1$.
\end{reptheorem}
\begin{proof}
In this case, we necessarily have $x^* < +\infty$. The analysis follows in a similar manner to the $\gamma \in (0,1)$ case, but for the value $x^* - G_M(n)$ instead of $G_M(n)$.

Notice that, by Lemma~\ref{lem:opt-dp-g}, we have
\begin{align*}
G_M(n) &= G_M(n-1) + \int^{x^*}_{G_M(n-1)} \prn{1 - F(u) \dif u} \\
&= G_M(n-1) + x^* - G_M(n-1) - \int^{x^*}_{G_M(n-1)} \prn{F(u) \dif u} \iff \\
x^* - G_M(n) &= \int^{x^*}_{G_M(n-1)} {F(u) \dif u}.
\end{align*}
Next, we analyze the integral on the right-hand side.

\begin{lemma}\label{lem:max-neg-g-int}
For $\gamma < 0$ and large $n$, we have
\[
\int^{x^*}_{G_M(n-1)} {F(u) \dif u} \approx \prn{1 - e^{-H(G_M(n-1))} \prn{1 - \frac{1}{1-\gamma}}} \prn{x^* - G_M(n-1)}.
\]
\end{lemma}
\begin{proof}
We have
\[
\int^{x^*}_{G_M(n-1)} {F(u) \dif u} = \int^{x^*}_{G_M(n-1)} {1 - e^{-H(u)} \dif u}.
\]
Let $y = e^{-H(x)}$, which implies that $\dif u = \prn{H^{\inv}\prn{-log{y}}}' \dif y$. Then,
\begin{align*}
\int^{x^*}_{G_M(n-1)} & {1 - e^{-H(u)} \dif u} = \int^0_{e^{-H(G_M(n-1))}} {(1 - y) \prn{H^{\inv}\prn{-log{y}}}' \dif y} \\
&= \brk{(1 - y) H^{\inv}\prn{-log{y}}}^0_{e^{-H(G_M(n-1))}} + \int^0_{e^{-H(G_M(n-1))}} {H^{\inv}\prn{-log{y}} \dif y} \\
&= x^* - G_M(n-1) \prn{1 - e^{-H(G_M(n-1))}} - \int^{e^{-H(G_M(n-1))}}_0 {H^{\inv}\prn{-log{y}} \dif y} \\
&= x^* - G_M(n-1) \prn{1 - e^{-H(G_M(n-1))}} + \int^{e^{-H(G_M(n-1)}}_0 {x^* - x^* \dif y} - \\
&\qquad - \int^{e^{-H(G_M(n-1))}}_0 {H^{\inv}\prn{-log{y}} \dif y} \\
&= \prn{1 - e^{-H(G_M(n-1))}} \prn{x^* - G_M(n-1)} + \int^{e^{-H(G_M(n-1))}}_0 {\prn{x^* - H^{\inv}\prn{-log{y}}} \dif y}.
\end{align*}
Next, let $z = \f{1}{y}$, which implies $\dif y = -z^{-2} \dif z$, and thus
\begin{align*}
\int^{x^*}_{G_M(n-1)} {1 - e^{-H(u)} \dif u} &= \prn{1 - e^{-H(G_M(n-1))}} \prn{x^* - G_M(n-1)} - \\
&\qquad - \int^\infty_{e^{H(G_M(n-1))}} {z^{-2} \prn{x^* - H^{\inv}\prn{log{z}}} \dif z}.
\end{align*}
Finally, notice that, for $\gamma < 0$, $g(z) = x^* - H^{\inv}\prn{log{z}} \in RV_\gamma$ by \cite{dehaan-ferreira-evt} (Corollary 1.2.10) and also that $1+\gamma < 1$ which, by the general form of Karamata's theorem (\cite{bingham-rv}, Theorem 1.5.11) for $\sigma = -2$, implies that
\begin{align*}
\int^\infty_{e^{H(G_M(n-1))}} {z^{-2} \prn{x^* - H^{\inv}\prn{log{z}}} \dif z} &\approx \frac{e^{-H(G_M(n-1)}\prn{x^* - H^{\inv}\prn{\log{e^{H(G_M(n-1))}}}}}{1 - \gamma} \\
&= \frac{e^{-H(G_M(n-1)}}{1 - \gamma} \prn{x^* - G_M(n-1)}.
\end{align*}
Therefore,
\[
\int^{x^*}_{G_M(n-1)} {F(u) \dif u} = \prn{1 - e^{-H(G_M(n-1)}\prn{1 - \frac{1}{1 - \gamma}}} \prn{x^* - G_M(n-1)}.
\]
\end{proof}

Now, let $\lambda'_M(n) = \frac{x^* - G_M(n)}{x^* - \mu_{n:n}}$. We have
\begin{align*}
\lambda'_M(n) &= \frac{x^* - G_M(n)}{x^* - \mu_{n:n}} \\
&\approx \frac{x^* - G_M(n-1)}{x^* - \mu_{n-1:n-1}} \frac{x^* - \mu_{n-1:n-1}}{x^* - \mu_{n:n}} \prn{1 - e^{-H(G_M(n-1))}\prn{1 - \frac{1}{1 - \gamma}}} \\
&\approx \lambda'_M(n-1) \prn{1 - \frac{\gamma}{n} + \ltlo{\frac{1}{n}}} \prn{1 - e^{-H(G_M(n-1))}\prn{1 - \frac{1}{1 - \gamma}}} \\
&\approx \lambda'_M(n-1) \prn{1 - \frac{\gamma}{n} + \ltlo{\frac{1}{n}}} \prn{1 - e^{-H(\lambda_M(n-1) \mu_{n-1:n-1})}\prn{1 - \frac{1}{1 - \gamma}}},
\end{align*}
where the second equality follows from Lemma~\ref{lem:max-neg-g-int} and the third equality from Lemma~\ref{lem:mu-ratio-approx}. Next, notice that, by \cite{correa-iid}, $G_M(n) \geq 0.745 \mu_{n:n}$, and thus
\[
\lambda'_M(n) - \lambda'_M(n-1) \leq \frac{x^* - 0.745 \mu_{n:n}}{x^* - \mu_{n:n}} - 1 = \frac{0.255 \mu_{n:n}}{x^* - \mu_{n:n}} \leq \frac{0.255 x^*}{x^* - \mu_{1:1}} \prod_{j = 2}^n \prn{1 - \frac{\gamma}{j}} \approx 0,
\]
where the second inequality follows from Lemma~\ref{lem:mu-ratio-approx} and the fact that $\mu_{n:n} \leq x^*$, and the asymptotic equality at the end follows from the fact that $\lim_{n \to \infty} \prod_{j = 2}^n \prn{1 - \frac{\gamma}{j}} = 0$, for any $\gamma \in (0,1)$. Thus,
\begin{align*}
\lambda'_M(n-1) \prn{1 - \prn{1 - \frac{\gamma}{n} + \ltlo{\frac{1}{n}}} \prn{1 - e^{-H(\lambda_M(n-1) \mu_{n-1:n-1})}\prn{1 - \frac{1}{1 - \gamma}}}} \approx 0,
\end{align*}
and thus
\[
e^{-H(\lambda_M(n-1) \mu_{n-1:n-1})}\prn{1 - \frac{1}{1 - \gamma}} \approx \frac{-\gamma}{n},
\]
Therefore,
\[
e^{-H(\lambda_M(n-1) \mu_{n-1:n-1})} \approx \frac{1-\gamma}{n} \iff H(\lambda_M(n-1) \mu_{n-1:n-1}) \approx -\log\prn{\frac{1-\gamma}{n}}.
\]
Taking the inverse of $H$, we obtain
\[
\lambda_M(n-1) \mu_{n-1:n-1} \approx H^{\inv\prn{-\log\prn{\frac{1-\gamma}{n}}}}.
\]
Since $H^{\inv}\prn{-\log{x}} = F^{\inv}(1 - x)$, we obtain
\[
G_M(n-1) \approx F^{\inv}\prn{1 - \frac{1-\gamma}{n}}.
\]
Therefore, by Lemma~\ref{lem:mu-approx}, we get
\[
\frac{x^* - G_M(n)}{x^* - \mu_{n:n}} \approx \frac{x^* - F^{\inv}\prn{1 - \frac{1-\gamma}{n+1}}}{x^* - F^{\inv}\prn{1 - \frac{1}{n}}}.
\]
Finally, using Lemma~\ref{lem:mult-quantile-approx}, we get
\[
\lambda'_M(n) \approx \frac{\prn{1 - \gamma}^{-\gamma}}{\Gamma(1 - \gamma) \prn{1+\frac{1}{n}}^{-\gamma}} \cdot \frac{x^* - F^{\inv}\prn{1 - \frac{1}{n}}}{x^* - F^{\inv}\prn{1 - \frac{1}{n}}} = \frac{\prn{1 - \gamma}^{-\gamma}}{\Gamma(1 - \gamma) \prn{1+\frac{1}{n-1}}^{-\gamma}}.
\]
Since $\lambda'_M(n) \to \lambda'_M$ as $n$ goes to infinity, we have
\[
\lambda'_M = \frac{\prn{1 - \gamma}^{-\gamma}}{\Gamma(1 - \gamma)}.
\]
However, this implies that
\[
\frac{x^*}{\mu_{n:n}} - \frac{G_M(n)}{\mu_{n:n}} \approx \frac{\prn{1-\gamma}^{-\gamma}}{\Gamma(1-\gamma)} \prn{\frac{x^*}{\mu_{n:n}} - 1}.
\]
Notice that for large $n$, the right-hand side goes to $0$, and thus it must be that
\[
\lim_{n \to \infty} {\frac{G_M(n)}{\mu_{n:n}}} = 1.
\]
Therefore, for $\gamma < 0$, we have $\lambda_M = 1$.
\end{proof}

\end{document}